\pdfoutput=1

\documentclass[a4paper,twoside]{report}
\usepackage[utf8]{inputenc}
\usepackage{amsthm}
\usepackage{amsmath}
\usepackage{amssymb}
\usepackage{stmaryrd} % \lightning
\usepackage{nicefrac} % beautiful fractions with \nicefrac
\usepackage{morefloats} % more marginpars

\usepackage{enumerate}
\usepackage{graphicx}
\usepackage{wrapfig}
\usepackage{tabularx}
\usepackage{booktabs}
\usepackage[numbers]{natbib} % must be before hyperref
\usepackage[ruled,vlined]{algorithm2e}
\usepackage{subfig}
\usepackage{color}
\usepackage{microtype} % abstaende und raender an buchstaben anpassen damits uniformer aussieht.
\usepackage[pagebackref]{hyperref}
\usepackage{breakurl}

\hypersetup{breaklinks=true}

\newcolumntype{L}{>{\raggedright\arraybackslash}X}%
\newcolumntype{C}{>{\centering\arraybackslash}X}%
\newcolumntype{R}{>{\raggedleft\arraybackslash}X}%

\usepackage{ifpdf}
\ifpdf 
\else 
\fi 

\renewcommand{\cite}{\citep}

\newcommand{\secref}[1]{\autoref{#1} (``\nameref{#1}'')}

\newcommand{\marginemph}[1]{\emph{#1}}%\marginpar{\small#1}}

\newcommand{\gr}{Golomb ruler}
\newcommand{\gro}{\textsc{Optimal Golomb Ruler}}
\newcommand{\grd}{\textsc{Golomb Ruler Decision}}
\newcommand{\grdu}{\textsc{Unary Golomb Ruler Decision}}

\newcommand{\grsm}{\textsc{Golomb Subruler}}
\newcommand{\grmd}{\textsc{Golomb Subruler Mark Deletion}}
\newcommand{\wacs}{\textsc{Weighted Antimonotone 2-CNF SAT}}
\newcommand{\chr}{\textsc{Characteristic Hypergraph Recognition}}

\newtheorem{definition}{Definition}[section]
\newtheorem{theorem}{Theorem}[section]
\newtheorem{lemma}{Lemma}[section]
\newtheorem{corollary}{Corollary}[section]
\newtheorem{rrule}{Reduction Rule}[section]
\newtheorem{construction}{Construction}[section]
\newtheorem{observation}{Observation}[section]

\newcommand{\decprob}[3]{%
  \begin{center}%
    \begin{minipage}{0.9\linewidth}%
      \textsc{#1}\\
      \textbf{Input:} #2\\
      \textbf{Question:} #3
    \end{minipage}%
  \end{center}%
}

\newcommand{\optprob}[3]{%
  \begin{center}%
    \begin{minipage}{0.9\linewidth}%
      \textsc{#1}\\
      \textbf{Instance:} #2\\
      \textbf{Task:} #3
    \end{minipage}%
  \end{center}%
}

\newcommand{\optprobnotitle}[3]{%
  \begin{center}%
    \begin{minipage}{0.9\linewidth}%
      \textbf{Instance:} #2\\
      \textbf{Task:} #3
    \end{minipage}%
  \end{center}%
}
\newcommand{\probsp}{\vspace{2mm}}  %TODO besserer weg?

\newcommand{\claP}{\text{P}}
\newcommand{\claNP}{\text{NP}}
\newcommand{\PeqNP}{$\claP = \claNP$}
\newcommand{\PneqNP}{$\claP \neq \claNP$}
\newcommand{\twonat}{\beta^\mathbb{N}}
\newcommand{\twonatapprox}{\beta^\mathbb{N}_{\sim}}

\renewcommand{\emptyset}{\varnothing}

\newenvironment{itemize+}[1][0]
  { \begin{itemize}
    \addtolength{\itemsep}{#1\baselineskip}
    \addtolength{\baselineskip}{#1\baselineskip} }
  { \end{itemize} }

\DeclareMathOperator{\length}{length}
\DeclareMathOperator{\bigO}{O}
\DeclareMathOperator{\bigOmega}{\Omega}

\begin{document}

\begin{titlepage}
\centering
\noindent \rule{\textwidth}{0.5pt}

\vspace{1em}
\Huge Algorithmic Aspects of\\ Golomb Ruler Construction

\normalsize\vspace{1.12\topsep}\large
%by

\normalsize\vspace{\topsep}\Large
\textsc{Manuel Sorge}\\
\rule{\textwidth}{0.5pt}

\vfill
\Large Studienarbeit\\
\normalsize \today

\vspace{\topsep}
  Supervision:\\
  Dr. Hannes Moser,\\
  Prof. Dr. Rolf Niedermeier, \\
  Dipl.-Inf. Mathias Weller

\vfill
\normalfont
\includegraphics[width=5cm,trim=0cm 4cm 0cm 0cm,clip=true]{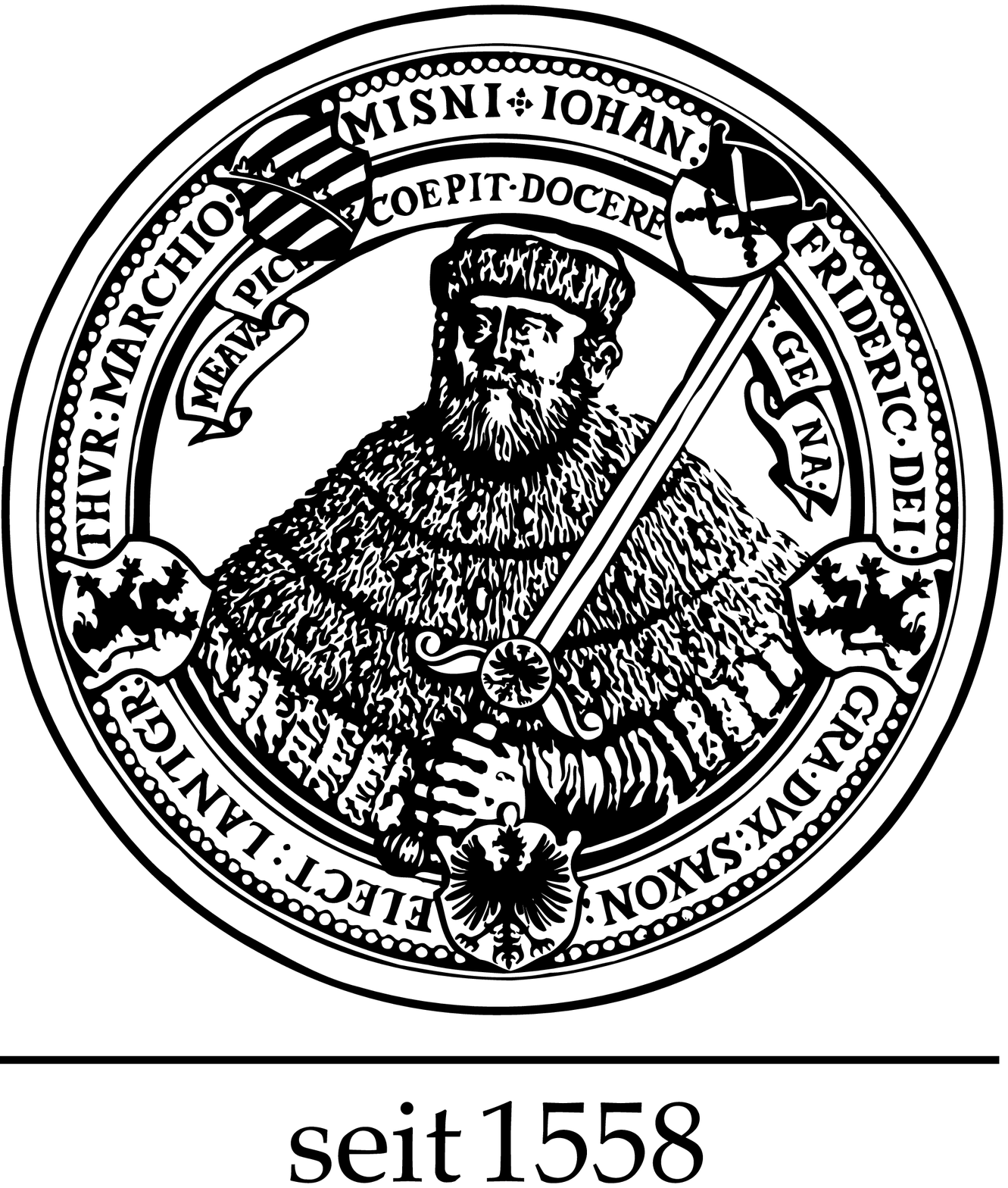}

\vspace{\topsep}
Friedrich-Schiller-Universität Jena\\
  Institut für Informatik\\
  Theoretische Informatik I / Komplexitätstheorie
\end{titlepage}

\begin{abstract}
  We consider \gr s and their construction. Common rulers feature marks at every unit measure, distances can often be measured with numerous pairs of marks. On \gr s, for every distance there are at most two marks measuring it. The construction of optimal---with respect to shortest length for given number of marks or maximum number of marks for given length---is nontrivial, various problems regarding this are NP-complete. We give a simplified hardness proof for one of them. We use a hypergraph characterization of rulers and \gr s to illuminate structural properties. This gives rise to a problem kernel in a fixed-parameter approach to a construction problem. We also take a short look at the practical implications of these considerations.
\end{abstract}

\tableofcontents

\chapter{Introduction}

%\section{Introduction}
\label{sec:intro}
A Golomb ruler is a specific type of ruler: Whereas common rulers have marks at every unit measure, a Golomb ruler only has marks at a subset of them. Precisely, the distance measured by any two marks on a Golomb ruler is unique on it. An example can be seen in \autoref{fig:golomb-common}. Golomb rulers are named after Professor Solomon Golomb. According to various sources \cite{Col03, Dim02}, he was one of the first to study their construction. % TODO Colannino rauskicken / was vom goloms

\gr s have various applications ranging from radio astronomy to cryptography. This explains the interest in computing \gr s that are particularly short for a given number of marks or have many marks when given a maximum length. Unfortunately, from a computational complexity point of view, some of the corresponding decision problems have been proven to be NP-complete, while little is known about other very natural problems.

Despite this, much effort has been made to compute short or dense \gr s and to prove them optimal. Various implementations of exhaustive searches have been given and discussed as well as heuristic and evolutionary approaches. A sophisticated project searches for \gr s through a distributed computer network, enabling users to donate idle computing time.

In this work, we give a short insight into the work that has been done in the field and we briefly consider two natural problems of unsettled computational complexity. We give a natural hypergraph characterization for rulers such that only \gr s correspond to a specific subset of the graphs. We then consider a construction problem that has been proven to be NP-complete. We give a simplified proof for this and then look at two natural parameterizations. For one of the parameterizations, we provide a fixed-parameter algorithm, and some heuristic improvements along with a cubic-size problem kernel that mainly follows from some structure that we observe in characteristic hypergraphs. Finally, we implemented an algorithm that uses the fixed-parameter approach and comment on our experimental results.

\begin{figure}
  \begin {center}
    \includegraphics{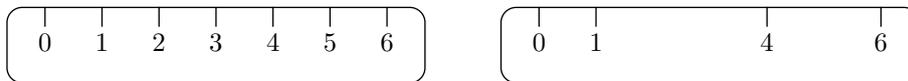}
    \caption{A common ruler (left) and a \gr{} (right) is shown. For every distance in \gr s, there is at most one pair of marks that measure this distance. For example, the distance one is only measured by the marks 0 and 1 on the \gr{} whereas this distance is measured by six pairs of marks on the common ruler. Both rulers measure every integer distance up to their length. Such rulers are called perfect. This, however, is a rare trait among \gr s, as one can proof \cite{Dim02}. }
    \label{fig:golomb-common}
  \end{center}
\end{figure}

\section{Origins and Applications}
\label{sec:originsappl}
According to \citet{Col03} and \citet{Dim02}, W. C. Babcock first discovered Golomb rulers while analyzing positioning of radio channels in the frequency spectrum. He investigated inter-modulation distortion appearing in consecutive radio bands~\cite{Bab53} and observed that when positioning each pair of channels at a distinct distance, then third order distortion was eliminated and fifth order distortion was lessened.

\citet{Ran93} lists other interesting applications, two of which we touch shortly here. In radio astronomy, arrays of radio telescopes are used to gather information about celestial bodies via interferometry. The telescopes are arranged in a single line, and information is extracted from difference measurements between two telescopes \cite{BBR74}. By placing them at the marks of a Golomb ruler, the number of these measurements and thus information gathered is maximized. This is a special case of a linear array. Linear arrays are also used in other related fields such as antennae construction. %TODO citation?

In computer communication networks \gr s can be used to simplify the message passing process. When allocating the node names corresponding to marks on a \gr, messages do not need to specify both origin and destination addresses. Since the differences between marks in a \gr{} are unique, the difference and the direction of arrival suffice to identify the origin and destination node \cite{BG77}.
% Golomb rulers are also applied at x-ray crystallography etc sdarticle.
% interessant auch hash probing: http://www.springerlink.com/content/ua5y8vwpd44v3nak/
\section{Preliminaries and Problem Definitions}
\label{sec:problemdef}

\reversemarginpar

We now gather a common knowledge base for our considerations. At first we will introduce rulers and \gr s, then go on to hypergraphs, some basic fixed parameter techniques and finally define some notation. We assume the reader to be familiar with basic mathematics and classic complexity theory. There are many recommendable books on complexity theory, see for example \cite{AB09, Pap94}.

\subsection{Rulers}
% The canonical definitions used throughout the literature are as follows: \cite all

\begin{definition}[Ruler]
  A \marginemph{ruler} is a set $R:=\{ m_i:1\leq i\leq n\}\subset \mathbb{Z}$ with~$m_i < m_{i + 1}$. The mark $m_i$~is called the \marginemph{$i$'th mark} on $R$. The ruler is said to have $n$ marks and~$|m_n-m_1|$ is called the \marginemph{length} of the ruler. We call a set~$R' \subseteq R$ a \marginemph{subruler} of $R$.
\end{definition}
\begin{definition}[\gr]
 A ruler $R = \{m_i : 1 \leq i \leq n\}$ is called \marginemph{\gr} if for every $d \in \mathbb{N} \setminus \{0\}$ there is at most one solution to the equation $d = m_i - m_j$, $m_i, m_j \in R$. %This is also called the Golomb trait!?
\end{definition}
It is easy to see that if a ruler $\{m_i : 1 \leq i \leq n\}$ is Golomb, so are the rulers~$\{m_ic : 1 \leq i \leq n\}$ and $\{m_i+c : 1 \leq i \leq n\}$ for a constant $c \in \mathbb{Z}$.

So for every \gr{} $R=\{ m_i:1\leq i\leq n\}$, there is a \gr{}~$R'$ with only positive marks, starting with the mark $0$ and having the same set of differences $\{m_i - m_j: m_i \neq m_j \wedge m_i,m_j \in R\}$. A ruler~$R'$ is also called the \marginemph{canonical form} of $R$ \cite{Dim02}. %TODO? mirror images
\begin{definition}[Perfect ruler]
  A ruler $R=\{ m_i:1\leq i\leq n\}$ is called \marginemph{perfect} if for every integer $1 \leq d \leq (m_n - m_1)$ there is at least one solution to the equation $d = m_i - m_j$, $m_i,m_j \in R$.
\end{definition}
It is not hard to see that there are no perfect \gr s with more than four marks \cite{Ran93}. In fact, every \gr{} with $n > 4$ marks has length greater than~$ n (n-1) / 2$.

This insufficiency leads to following problem:
% Given $n$, what is the minimum length of a \gr{} with $n$ marks?
\probsp{}
\optprob{\gro}{$n \in \mathbb{N}$.}{Find a \gr{} with $n$ marks and minimum length.}
\probsp{}
A natural decision problem (as posed by \citet{MP08}) is:
\probsp{}
\decprob{\grd}{$n, D \in \mathbb{N}$.}{Is there a \gr{} with at least $n$ marks and length at most $D$?}
\probsp{}
\begin{definition}[Optimal \gr{} and the function $G(n)$]
  A \gr{} with $n$~marks is called \marginemph{optimal} if it is of shortest possible length.
  For every $n \in \mathbb{N}$ the function \marginemph{$G(n)$} is defined as the length of an optimal \gr{} with $n$ marks.
\end{definition}
No closed form expression is known for $G(n)$ \cite{Dim02}, and even using massive computational power, $G(n)$ to date is only known for $n \leq 26$ \cite{DistribStat}.

%\section{Variants and related Problems}
%The \mrla{} problem is the following:
%\probsp{}
%\optprob{\mrla}{$D \in \mathbb{N}${Find a perfect ruler of length $D$ with the minimum of redundant 

\subsection{Graphs and Hypergraphs} %Hypergraphs are a generalization of graphs. 
A \marginemph{hypergraph}~$H$ is a tuple~$(V,E)$, where $V$ is a finite set and $E$ is a family of sets~$e$ such that~$e \subseteq 2^V \setminus \{\emptyset\}$, where $2^V$ denotes the powerset of~$V$. The elements of~$V$ are called \marginemph{vertices} and the elements of~$E$~\marginemph{edges}. Two vertices are called \marginemph{adjacent}, if there is an edge that contains both of them. A vertex~$v \in V$ and an edge~$e \in E$ are called \marginemph{incident}, if~$v \in e$. 

Two hypergraphs $H=(V,E)$ and $H'=(V',E')$ are called \marginemph{isomorphic} and we write~$H \cong H'$, if there is a bijection~$\phi:V \rightarrow V'$ such that the following holds: $$\{v_1, ..., v_i\} \in E \Leftrightarrow \{\phi(v_1), ..., \phi(v_i)\} \in E'$$

The hypergraph $H'=(V',E')$ is called subhypergraph or short \marginemph{subgraph} of a hypergraph~$H=(V,E)$, if~$V' \subseteq V$ and $E' \subseteq E$. The hypergraph $H$ is then also called a \marginemph{supergraph} of $H'$. The hypergraph $H[V'] := (V', E')$ is called (vertex) \marginemph{induced subgraph} of a hypergraph $H=(V,E)$, if~$V' \subseteq V$ and $E'$~contains every set~$e$, such that~$e \in E$ and $e \subseteq V'$. A hypergraph $M_H$ is called a \marginemph{minor} of a hypergraph $H$, if $M_H$ can be obtained from $H$ by removing vertices, removing edges and contracting edges. Contracting an edge $e$ means to delete $e$ and every vertex contained in $e$ and to introduce a new vertex $v$ that is added to every edge that was incident to a vertex in~$e$. 

The hypergraph~$H$ is called \marginemph{simple}, if~$e \neq \hat{e}$ for every~$e, \hat{e} \in E$. The hypergraph~$H$ is called \marginemph{$d$-uniform} for an integer~$d$, if~$|e| = d$ for every~$e \in E$. 

We draw a hypergraph by drawing points for every vertex, drawing curves for every edge~$e, |e| > 2$, encircling all vertices in that edge, drawing straight lines between the contained vertices for edges $e, |e| = 2$ and drawing loops at the contained vertex for edges $e, |e| = 1$.

A \marginemph{graph} is a simple 2-uniform hypergraph. A graph is called \marginemph{complete} or a \marginemph{Clique}, if all possible edges are present in the graph. A graph~$G=(V_1 \cup V_2,E)$ is called \marginemph{bipartite}, if~$V_1 \cap V_2 = \emptyset$ and there is no~$e \in E$ such that~$e \subseteq V_1$ or~$e \subseteq V_2$. A bipartite graph~$G=(V_1 \cup V_2, E)$ is called \emph{complete} or a~$K_{|V_1|,|V_2|}$ if it contains all possible edges and still maintains its bipartite property. A graph is called \marginemph{planar}, if it can be embedded in the plane, i.e. it can be drawn in a plane such that the drawings of the edges intersect only in their endpoints. A well known theorem by Kuratowski states that a graph is planar, if and only if it does not contain a clique with five vertices or a $K_{3,3}$ as minor \cite{We00}.

The \marginemph{incidence graph} of a hypergraph~$H = (V,E)$ is defined as the bipartite graph~$I = (V \cup E, E')$ with $E' = \{ \{v, e\} : v \in V \wedge e \in E \wedge v \in e\}$. A hypergraph is called planar, if its incidence graph is planar.

\subsection{Basic Fixed-Parameter Techniques and Complexity Theory}

Many natural problems are NP-hard and thus are believed not to be solvable within running time bounded by a polynomial function. However, in practise the phenomenon can be observed that some instances of NP-hard problems in fact can be solved within reasonable time. %TODO cite?
This is because in classic computational complexity mostly worst case running times depending only on the input length are contemplated. The exponential worst case running times notwithstanding it is possible to identify structures that can be exploited by algorithms in some problems. Fixed-parameter algorithmics can be seen as the approach to find efficient algorithms not only with respect to input length but also to such structures---called parameters. We only recapitulate some very basic definitions and techniques here, for more on the topic see e.g. \cite{Nie06, DF99, FG06}.

A \marginemph{parameterized problem} is a language $L\subseteq \Sigma^* \times \Sigma^*$, where $\Sigma$ is a finite alphabet. The second component is called the \marginemph{parameter} of the problem. In our problems the parameter will always be a nonnegative integer and therefore we restrict this definition to languages $L \subseteq \Sigma^* \times \mathbb{N}$ in this paper. A parameterized problem $L$ is \marginemph{fixed-parameter tractable} with respect to the parameter $k$ if there exists an algorithm that decides $L$ in $f(k)p(n)$ time, where $f:\mathbb{N} \rightarrow \mathbb{N}$ is a computable function only depending on $k$ and $p$ is a polynomial.

Bounded \marginemph{search trees} are a standard way to classify a problem as fixed-parameter tractable. Search trees are a way of systematic exhaustive search. The strategy is to find a small part of the input in polynomial time, such that at least one element of this part has to be in the solution. Then we branch into all cases of choosing one element of this part and recurse until a solution is found. The graph with the (recursive) calls of the algorithm as nodes and an edge between two nodes, if one was called by the other is called the search tree. If in every node time is spent that is bounded by a polynomial in the input length and if we can bound the number of succeeding recursive calls at one node and the height of the tree by a function that depends only on the parameter, we obtain a fixed-parameter algorithm. If the algorithm has an input of size $s$ and branches into recursively solving instances of sizes~$s - d_1, ..., s - d_i$, then $(d_1, ..., d_i)$ is called the \marginemph{branching vector} of this recursion. 
Since search tree algorithms often terminate early and are easily parallelized, this technique has applications in practise. %TODO Cite

Let $L$ be a parameterized problem. A \marginemph{reduction rule} is a mapping of instances~$(I,k)$ to instances~$(I',k')$ such that the following conditions hold: First, $(I,k) \in L \Leftrightarrow (I',k') \in L$, which is also called \marginemph{correctness} of the rule. Second, $\length(I') \leq \length(I)$ and $k' \leq k$. An instance is called \marginemph{reduced} with respect to a reduction rule, if the rule cannot be applied to the instance anymore. That is, the image of the instance under the reduction function is the same as the instance itself. A reduction to a \marginemph{problem kernel} is a reduction rule that can be computed in $\bigO((\length(I))^c)$~time for some constant~$c$ such that $\length(I') \leq g(k)$, where $g: \mathbb{N} \rightarrow \mathbb{N}$ is computable and depends only on $k$. The function~$g$ is also called the \marginemph{size} of the problem kernel.

Let $L, L'$ be two parameterized problems. A \marginemph{parameterized reduction} from~$L$ to $L'$ is a function~$r:\Sigma^* \times \mathbb{N} \rightarrow \Sigma^* \times \mathbb{N}, (I, k) \mapsto (I',k')$ such that $r$ is computable in $f(k)p(\length(I,k))$ time, $(I,k) \in L \Leftrightarrow (I',k') \in L'$ and $k'$ depends only on $k$. Here, $f$ is a computable function depending only on $k$ and $p$ is a polynomial.

A parameterized problem $L$ belongs to the class \marginemph{W[t]} if there is a parameterized reduction from $L$ to a weighted satisfiability problem for the family of circuits of weft---the maximum number of gates with unbounded fan-in on an input output path---at most $t$ and depth at most some function of the parameter~$k$. For an introduction to the W-hierarchy see \cite{FG06}. We only use the fact that W[1]-hard problems are assumed not to be fixed-parameter tractable. A parameterized problem can be shown to be \marginemph{W[1]-hard} by giving a parameterized reduction from a W[1]-complete problem. For example, the following problem is W[1]-complete with respect to parameter $k$.
\decprob{Independent Set}{A graph $G=(V,E)$ and an integer $k$.}{Is there a vertex subset $S \subseteq V$ such that $k \leq |S|$ and $G[S]$ contains no edges?}
%TODO? parameterized-problem umgebung

\subsection{Conventions}
We denote the number of vertices of a (hyper-)graph by~\marginemph{$n$} and the number of edges by~\marginemph{$m$}, where it is not ambiguous. Also, we denote an edge of size $d$ in hypergraphs by \marginemph{$d$-edge}. Hypergraphs whose edges are of size exactly three or four play a major role in our considerations. We call these graphs~\marginemph{3,4-hypergraphs}.

In proofs, we use the lightning-symbol~$\lightning$ to indicate an exposed contradiction and thus a finished case of a reductio ad absurdum.

\section{Previous Work}
\label{sec:pastresearch}

In this section, we give a quick overview of the research that has been done in the \gr{} field in the past.

\subsection{Golomb Ruler Construction and Bounds on $G(n)$}
\label{sec:prbounds}

\citet{Dim02} analyzed the \gr{} problem with respect to the older so called Sidon set problem. Informally, a Sidon set is a finite set~$S \subset \mathbb{N}$ such that the sum of any two elements of $S$ is distinct. %TODO? b2-sequences
It can easily be shown that this definition is equivalent to the definition of the \gr s. Similarly to \gr s, there is an optimization problem for Sidon sets that asks the following: Given~$n$, what is the maximum cardinality of a~Sidon~set~$S \subseteq \{1, ..., n\}$?

This problem has been studied extensively. Lower bounds have been given by \citet{Lin69} (as claimed in \cite{Dim02}) and \citet{ET41} via construction strategies. These bounds then have been applied to \gr s by \citet{Dim02}, yielding the lower bound $G(n) > n^2 - 2n \sqrt{n} + \sqrt{n} - 2$.

Other construction strategies for Sidon sets and thus \gr s have been given indirectly. \citet{Sin38} discovered a method for generating a set of $q + 1$ residues modulo $q^2 + q + 1$ that form a \gr, $q$ being a prime power. \citet{Dim02} and \citet{Rus93} claim that \citet{Bos42} found a strategy to do the same for $q$ residues modulo $q^2 - 1$ and \citet{Rus93} found one for $p - 1$ integers such that their pairwise sums are all different modulo $p(p-1)$, $p$ being a prime number. The former implies that $G(q) \leq q^2 - 1$ for prime powers $q$.

Still, for all integers the best known upper bound is $G(n) \leq 2n^3 + n$. This is due to a relatively simple construction described in \cite{Dim02}.

Erd{\H{o}}s conjectured an upper bound to be $G(n) \leq n^2 + c$, with $c \in \mathbb{R}$ being a constant. \citet{Dim02} has computed relatively short \gr s and thus showed with computer aid that $G(n) < n^2$ for $n \leq 65,000$.
%TODO? Bose-chowla? Erdos cite? 

\subsection{Computational Approaches}
\label{sec:prcomput}

Numerous efforts have been made to compute \gr s and prove their optimality with computer aid. We can only cover some selected cases here.

\citet{Ran93} has developed some exhaustive search algorithms, reviewed their empirical running time, and provided a parallel implementation. He has been able to compute and prove optimal \gr s with 17 and 18 marks in 1993. Using a cluster of one Sun SparcServer 1000 and twelve Sun SparcClassic workstations, it took about 840 hours to complete the run for the 18-mark ruler.

\citet{Distrib}---founded in 1997---attacks the \gr{} problem through a globally distributed computer system. A client is provided, which enables the user to donate idle computing time to the project. The participants were able to show that previously shortest known \gr s of length 24 through 26 were indeed optimal. With 124,387 participants during the run for the 25-mark \gr, it took 3,006 days to find the result. The run for the 26-mark \gr{} allegedly took only 24 days. Unfortunately, the source code of the projects key elements are not publicly available because of security reasons. %So it is not clear, what is done exactly.

\gr s have also raised some interest in the field of evolutionary algorithms and other heuristic techniques. For example, see articles by \citet{TPC07, CDFH07, PTC03} and \citet{CF05}.

\subsection{Complexity Theory}
\label{sec:prcomplex}

Surprisingly, given the number of implementations, seemingly little is known about the computational complexity of \grd{} and \gro{}. % and problems related to the construction of Golomb~Rulers~
See recent publications by \citet{MP08} and \citet{MaY08}.

\citet{MP08} have focussed on the construction of \gr s and proved the following problems to be NP-complete:
\probsp{}
\decprob{\grsm}{A finite set $S \subseteq \mathbb{N}$ and $n\in\mathbb{N}$.}{Is there a \gr{} $S' \subseteq S$ with at least $n$ marks?}
%\probsp{}
\decprob{Golomb Ruler Sum}{A finite set $T \subseteq \mathbb{N}$ and $D, n \in \mathbb{N}$.}{Are there elements $t_1, ..., t_i \in T$ such that $n - 1 \leq i$ and the ruler $\{ \sum_{j=1}^kt_j : 1 \leq k \leq i \} \cup \{0\}$ is a \gr{} of length equal to~$D$?}
%\probsp{}
\decprob{Golomb Ruler Subset Distances}{A finite set of interval lengths $T \subseteq \mathbb{N}$ and $n \in \mathbb{N}$.}{Is there a \gr{} $R = \{m_i : 1 \leq i\leq n\}$ such that $\{|m_j - m_i | : 1 \leq i < j \leq n\} \subseteq T$?}
\probsp{}
\citet{MaY08} have reduced \gro{} to a problem called \textsc{Seed Optimization}. Unfortunately, they also note that the complexity of this problem is unknown.

However, there has also been some research on intuitively related problems, i.e., \textsc{Difference Cover} and \textsc{Turnpike}: A set $\Delta \subseteq \mathbb{N}$ is called a \emph{difference cover} for a set $Y \subset \mathbb{N}$ if for each $y \in Y$ there exist at least two elements $a,b \in \Delta$, such that $y = a - b$. \citet{MP06} have proven that a polynomial time algorithm for the following problem would imply $\text{P}=\text{NP}$.
\probsp{}
\optprob{Minimum Difference Cover}{A set $Y \subset \mathbb{N}$.}{Find the minimum cardinality $\Delta \subset \mathbb{N}$, such that $\Delta$ is a difference cover for $Y$.}
\probsp{}
The \textsc{Turnpike} problem is defined as follows:
\probsp
\decprob{Turnpike}{A multiset of $n(n - 1)/2$ integer distances.}{Is there a set of $n$ points in $\mathbb{N}$ with the given distances?}
\probsp
\citet{MP08} note that the complexity of \textsc{Turnpike} is unknown. Variants and special cases of it have been studied, some of which have been proven to be in P, some to be pseudo-polynomial time solvable and others to be NP-complete. See references in \cite{MP08}.

% \subsection{Graph Theory}
% Golomb Rulers also have a link to Graph Theory. \citet{Bea01} has observed that a $K_n$ is graceful (it has a graceful labeling) if and only if there is a perfect Golomb Ruler with $n$~marks. Furthermore, he has been able to prove that an induced $K_n$ in a graceful graph induces a Golomb Ruler of $n$ marks. 
% %cite graceful labeling.

\chapter{Algorithms and Complexity}

\citet{MP08} provided insight into the computational complexity of constructing \gr s. However, there are still white spots in the map of complexity of problems related to \gr s. In \autoref{sec:grocomplex} we consider one central white spot, and try to explain why it still has not been settled.

In \autoref{sec:hgraphchara} we introduce the notion of characteristic hypergraphs for rulers. This new technique helps illustrating problems related to \gr s  and we also gather some structural insights based on these graphs. The characterization serves as base for our considerations in the succeeding sections:

\autoref{sec:simplnphardness} contains an alternative and simplified proof for the NP-hardness of \grsm{}. (The original proof has been given by \citet{MP08}.) We also consider the fixed-parameter tractability of \grsm{} with two natural parameters in \autoref{sec:fpt}. We give a positive result for one of them and prove a cubic-size problem kernel.

\section{Notes on the Complexity of \gro}
\label{sec:grocomplex}

Recall the definitions of \gro{} and \grd:

\probsp{}
\optprob{\gro}{$n \in \mathbb{N}$.}{Find a \gr{} with $n$ marks and minimum length.}
\decprob{\grd}{$n, D \in \mathbb{N}$.}{Is there a \gr{} with at least $n$ marks and length at most $D$?}
\probsp{}

Although some authors \cite{Dim02, SHL95} believe that \gro{} is computationally hard, to date there is no proven evidence. Even \citet{MP08}, who focussed on the complexity of problems related to \gr s, do not state a conjecture on whether \grd{} is NP-hard or whether it is in NP. We now discuss why these questions seem to be difficult to answer.

At first, observe that the encoding of the maximum length $D$ and the minimum number of marks $n$ has a heavy impact on the complexity of \grd{}. This is due to the following.

\begin{theorem}[\citet{Ber78}]
  \label{the:berman}
  If there is an NP-complete language L and a polynomial function~$f: \mathbb{N} \rightarrow \mathbb{N}$ such that the following statement holds, then \PeqNP. $$\forall l \in \mathbb{N} : |\{x \in L: \length(x) \leq l\}| \leq f(l) $$
\end{theorem}

In \grd, let $D$ and $n$ be encoded with binary alphabet and let \grdu{} denote the same problem, but with the input encoded with unary alphabet. In every unary language, there are at most $l + 1$ words of length at most $l$. Therefore, the number of positive instances of length at most $l$ in \grdu{} is clearly bounded by a polynomial function in $l$. Thus, by \autoref{the:berman} and under the assumption that \PneqNP{}, \grdu{} cannot be NP-complete. The problem clearly lies in NP, because the trivial certificate---a \gr{} satisfying the conditions---is of length $\bigO(n \log{D})$, which is polynomial in~$n$ and~$D$, so it cannot be NP-hard.

For \grd, that is, using binary encoding of the input, it is not even clear whether it is in NP. This is due to the fact that, if it is in NP, then there must be a certificate of polynomial size for every instance. As noted above, the trivial certificate is of length $\bigO(n \log{D})$, which now is exponential in the input length $(\log{n}+\log{D})$. Obviously, this does not imply that such a certificate cannot exist, but it seems intuitively plausible that it comprises a \gr{} or some notion of the positions of its marks. To encode this into a word of size polylogarithmic in $D$ and $n$ seems to be a difficult task.

Because the number of marks $n$ is the culprit to the exponentiality of the trivial certificate here, it might seem obvious to try and encode only $n$ in unary and give $D$ in binary. However, this makes \autoref{the:berman} applicable again: 

\begin{observation}
  \label{obs:nontriviallanguages}
  Let $L$ be a language and let $f$ be a polynomial-time computable function such that either $f(w)=\perp$ or $f(w)=M$, where $M$ is the coding of a Turing machine that decides in polynomial time whether $w \in L$ or not. %, if such a machine is known.
  Furthermore, let $L_{nt} = \{ w \in L : f(w) = \perp\}$. Then $L$ and $L_{nt}$ are polynomial-time equivalent.
\end{observation}
\begin{proof}
  It is clear that the language $L_{nt}$ can be decided in polynomial time with~$L$, because~$L_{nt} \subseteq L$. 

  To reduce $L$ to $L_{nt}$, one simply computes $f(w)$ for $w \in L$. If $f(w) = \perp$ then $w \in L_{nt}$ and we can output $w$. Otherwise, simulate $f(w)=M$ with input $w$ and output $w_{no}$ or $w_{yes}$ if $M$ does not or does accept, respectively, where $w_{no} \notin L$ and $w_{yes} \in L$ (we can code a constant number of such words in the coding of the Turing machine computing this reduction).
\end{proof}
The interpretation of \autoref{obs:nontriviallanguages} is that any problem $L$ is polynomial-time inter-reducible with a problem $L_{nt}$ which contains only ``non-trivial instances'', that is, all instances of $L$ except those that are known to be decidable in polynomial time and that can be classified to be so in polynomial time. Therefore, for the sake of complexity classification, we can exclude trivial instances in every problem. 

We know from \secref{sec:prbounds} that the length of an optimal \gr{}  with $n$ marks is upper-bounded by $t \in \bigO(n^3)$. Thus, for a given number $n$ of marks, there are only polynomially in $n$ many lengths $D$ that form instances which can not trivially be checked in polynomial time. Out of the words of length at most $l$ in the language induced by \grd{} with $n$ in unary, $D$ in binary and without trivial instances, there can be at most $l + 1$ words that represent distinct values of $n$ and thus, there can be at most $t(l+1)$ words of length at most $l$. This means that if $n$ is given in unary in \grd{} and $D$ is encoded binarily, a language defined by this problem that excludes trivial instances again satisfies the condition of \autoref{the:berman}.

%TODO cite many np=complete problems have a notion ...
Also, since many NP-complete problems have a notion of efficient self-reduction, an approach to get a hint on the complexity of \grd{} would be to search for such a procedure. Self-reduction is a procedure to compute an optimal solution for an optimization problem using an oracle for the corresponding decision problem. Unfortunately, \citet{MP08} note that it seems difficult to efficiently compute a valid \gr, given an oracle for \grd.

\section{Characterizing Golomb Rulers through Hypergraphs}

\label{sec:hgraphchara}

In this section we provide a simple hypergraph characterization of \gr s and consider structural properties of the implied hypergraphs. The characterization serves as base for considerations in the succeeding sections.

\subsection{Hypergraph Construction}
\label{sec:hgraphconstr}

\begin{figure}
  \begin{center}
    \includegraphics{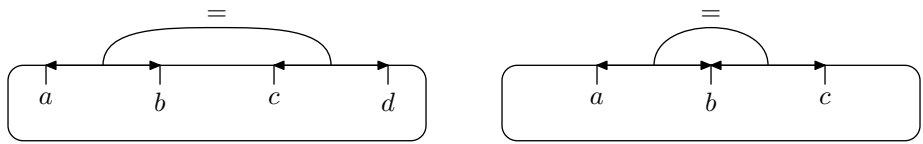}
    \caption{Two rulers with the marks $a,b,c$ and $d$, respectively. To the left, we see that the marks $a$ and $b$ measure the same distance as $c$ and $d$. We consider this to be a conflict with respect to \gr s and model it as an edge~$\{a, b, c, d\}$ in the corresponding hypergraph. To the right we see a degenerated form of a conflict which leads to an edge with only three vertices.}
    \label{fig:distancesedges}
  \end{center}
\end{figure}

We start with a ruler $R \subseteq \mathbb{N}$ and construct a hypergraph with edges consisting of either three or four vertices. Every mark corresponds to one vertex in the graph and every edge to one conflict, i.e., to a distance that is measured by two pairs of marks (see \autoref{fig:distancesedges}). Constructing such a graph can be done by simply iterating over every three and four-tuple of marks and checking whether the marks in the tuple have equal distances. See also \autoref{alg:hgraphconstr}. 
\begin{algorithm}%[H]
  \KwIn{A ruler $R \subset \mathbb{N}$.}
  \KwOut{A hypergraph $H_R=(R,E)$.}
  
  Start with an empty hypergraph $H$\;
  Let $R$ be the set of vertices in $H$\;
  \For{$(a, b, c, d) \in R^4, |\{a, b, c, d\} | = 4$}{
    \lIf{$|a - b| = |c - d|$}{add the edge $\{a, b, c, d\}$ to $H$}}
  \For{ $(a, b, c) \in R^3, |\{a, b, c\} | = 3$}{
    \lIf{$|a - b| = |b - c|$}{add the edge $\{a, b, c\}$ to $H$}}
  \Return $H$\;
  \SetAlgoRefName{Hyper\-graph\-Construction}
  \caption{Constructing a characteristic hypergraph for a given ruler}
  \label{alg:hgraphconstr}
\end{algorithm}
\begin{definition}
  In the following sections, we denote the hypergraph constructed from the ruler $R$ by $H_R=(R,E)$ and call it \emph{characteristic hypergraph of $R$}. 
\end{definition}
The construction of $H_R$ is clearly computable in $\bigO(|R|^4)$.
\begin{lemma}
  Let $R$ be a ruler and $H_R=(R,E)$ the corresponding characteristic hypergraph constructed from $R$ using \autoref{alg:hgraphconstr}. Then $R$ is a \gr{} if and only if $E = \emptyset$.
\end{lemma}
\begin{proof}
  We show that both directions of the following equivalent statement hold: The ruler $R$ is not a \gr{} if and only if $E \neq \emptyset$.

  First assume that $E$ is not empty. Thus, there is either an edge $\{a, b, c\}$ or  $\{a, b, c, d\}$ in $E$ and the corresponding equations $|a - b| = |c - d|$ or $|a - b| = |b - c|$, respectively, hold. In other words, $R$ contains a number of marks that have pairwise equal distances. That means $R$ is not a \gr{}.

  Now assume that $R$ is not a \gr{}. Thus, there are two equal differences~$|a - b| = |c - d|$ for some $a,b,c,d \in R$, where at least $a, b$ and $d$ are pairwise not equal. This means that $\{a, b, c, d\} \in E$ and $E \neq \emptyset$.
\end{proof}
We can improve the running time of the construction to $\bigO(|R|^3)$ using a different approach. Instead of simply verifying every possible tuple, one can look at the distances between marks present in the ruler and examine which of them lead to edges in the graph. \ref{alg:hgraphconstrimpr} is a description of such an algorithm.

In this algorithm we use an auxiliary map $M$ to keep track of pairs of marks that measure specific distances. At first, we fill up this map: The first two loops iterate over distances present in $R$ and add every pair of vertices to the entry in $M$ corresponding to their distance. The map $M$ then contains for every necessary distance in $R$ a list with all pairs of marks that measure this distance. In the second step, we add the edges to the designated characteristic hypergraph~$H$: The last three nested loops again iterate over distances present in the ruler and simply add an edge to $H$ for every pair of marks that measure this distance.
\begin{algorithm}%[H]
  \LinesNumbered

  \KwIn{A ruler $R \subset \mathbb{N}$.}
  \KwOut{A hypergraph $H_R=(R,E)$.}
  Start with an empty hypergraph $H$\;
  Let $R$ be the set of vertices in $H$\;
  Create an empty map $M$ that maps integers to lists\;
  $\delta_{max} \leftarrow \max\{x:x \in R\} - \min\{x:x \in R\}$\;
  \For{$i \in R$}{
    \For{$j \in R, i < j \leq i + \nicefrac{\delta_{max}}{2}$}{
      Add $(i,j)$ to the list mapped to $j - i$ in $M$\;}}

  \For{$i \in R$}{
    \For{$j \in R, i < j \leq i + \nicefrac{\delta_{max}}{2}$}{
      \For{$(k, l)$ in the list mapped to $j - i$ in $M, j \leq k$}{
%        If not present and i,j =/= k, l
        Add the edge $\{i, j, k, l\}$ to $H$\;}}}
  \Return $H$\;
  \SetAlgoRefName{Hyper\-graph\-Construction\-Improved}
  \caption{Constructing a characteristic hypergraph for a given ruler}
  \label{alg:hgraphconstrimpr}
\end{algorithm}
\begin{figure}%[5]{L}{0.475\textwidth}
  \begin{center}
    \includegraphics{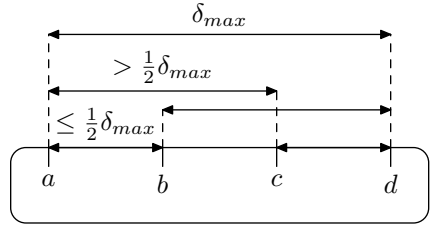}
    \caption{Shown is a ruler with four marks $a, b, c, d$. The pair $a, c$ measures a distance that is above half the maximum distance $\delta_{max}$ measurable by marks on the ruler. The pair $b, d$ measures the same distance as $a, c$. However, because they measure such long distances, the measurements must overlap. By~\autoref{lem:nonoverlap}, there also is another distance (measurable without overlap) that also leads to this conflict. This distance must of course be below $\nicefrac{\delta_{max}}{2}$.}
    \label{fig:longdistances}
  \end{center}
\end{figure}

For the correctness of this algorithm we first make some auxiliary observations:

\begin{lemma}
  \label{lem:nonoverlap}
  Every edge in a characteristic hypergraph is due to two pairs of marks that measure the same distance and the measurements do not overlap.
\end{lemma}
\begin{proof}
  Assume that there are four marks $a < b < c < d$ such that the following equation holds $$c - a = d - b =: \delta$$ That is, the measurements of the pairs $a, c$ and $b, d$ overlap. Then subtracting the overlap $c - b$ from the distance $\delta$ gives $b - a = d - c$, which implies measurements that do not overlap.
\end{proof}
\begin{corollary}
  \label{cor:longdistances}
  Every edge in the characteristic hypergraph of the ruler $R$ is due to a distance that is at most half the maximum distance measurable by marks on $R$ (see \autoref{fig:longdistances}). 
\end{corollary}
\begin{lemma}
  \autoref{alg:hgraphconstrimpr} constructs a characteristic hypergraph for its input ruler.
\end{lemma}
\begin{proof}
  Assume that there are marks $a < b \leq c < d$ in $R$ such that  $\{a, b, c, d\}$ has not been added to $H$ by the algorithm although the relation~$b - a = d - c \leq \nicefrac{\delta_{max}}{2}$ holds. We can assume the properties of the marks because of \autoref{lem:nonoverlap} and \autoref{cor:longdistances}.

  It is clear that $(a, b)$ and $(c, d)$ have been added to $M$ in lines 5 through 7. Then at some point in the execution of the algorithm, $a$ must occur in line 8 and $b$ in line 9. Because $(d, c)$ has been added to the list mapped to $b-a$ in $M$, the edge $\{a, b, c, d\}$ then is added to $H$. $\lightning$
\end{proof}
\begin{lemma}
  \autoref{alg:hgraphconstrimpr} runs in $\bigO(|R|^3)$~time.
\end{lemma}
\begin{proof}
  Obviously the running time is mainly dependent on the last three nested loops in lines 8-10. The two outer loops each iterate at most $|R|$ times. 
Retrieving the list from $M$ in line 10 can be done in $\bigO(\log(|R|))$~time, using red-black trees as implementation for $M$ \cite{CLRS01}. The iteration of the innermost loop in line 10 is bounded by a term in $\bigO(|R|)$, because any fixed distance~$\delta$ between two marks on the ruler $R$ can occur at most $2|R|$ times: $\delta$ can be measured at most two times by one mark with any other mark. Adding the edge to the hypergraph is possible in $\bigO(1)$ using an incidence graph representation of $H$ and adjacency lists for example and thus the running time is in $\bigO(|R|^3)$.
\end{proof}

Note that we omitted long distances in the loop-headers in lines 6 and 9. However, the omission does not influence the asymptotic upper bound on the running time. This is a heuristic trick that could prove useful in practice.

\begin{lemma}
  The upper bound on the running time of \autoref{alg:hgraphconstrimpr} is tight.
\end{lemma}

\begin{proof}
  There are characteristic hypergraphs that contain $\bigOmega(|R|^3)$~edges: This holds for graphs constructed from rulers whose marks are taken consecutively from $\mathbb{N}$. Consider the ruler $R = \{0, 1, ..., n\} \subset \mathbb{N}$ for even~$n$; obviously every distance up to $n$ is measured by marks in~$R$. Let $\delta \leq \nicefrac{n}{2}$ be a fixed distance measured. How many possibilities are there to choose two pairs of marks such that they both measure $\delta$ and the measurements do not overlap? One can place one pair leftmost on the ruler, count every possible placement of the other pair to the right, then move the first one to the right by one and iterate. Summing over every distance up to $\nicefrac{n}{2}$ this gives a lower bound on the number of edges in~$H_R$:

\begin{equation*}
  \sum\limits_{\delta=1}^{\nicefrac{n}{2}}\sum\limits_{j=0}^{n - 2\delta}\sum\limits_{k=j+\delta}^{n-\delta}1 \in \bigOmega(n^3)
\end{equation*}

No edge is counted twice here, because if there is an edge due to two different distances $1 \leq \delta_1 < \delta_2$, the measurements of $\delta_2$ overlap: let $a < b$ and $c < d$ both measure $\delta_1$ with $b \leq c$. Where can $\delta_2$ be placed? Certainly not between $a,b$ and the other pair and also not between the pairs $a, d$ and $b,c$. It can only be measured by both of the pairs $a,c$ and $b,d$ and thus the measurements overlap. However, in the above construction, we are only counting non-overlapping measurements.
\end{proof}

In conclusion, the following theorem now readily follows:

\begin{theorem}
  There is a hypergraph characterization for rulers such that \gr s and only \gr s correspond to hypergraphs without edges. The characteristic hypergraph for a ruler $R$ can be computed in time $\bigO(|R|^3)$ and this bound is also tight.
\end{theorem}

\subsection{Notes on the Structure of the Characteristic Hypergraphs}
\label{sec:hgraphstruc}

At first, notice that the set of hypergraphs that can be constructed from rulers as in \secref{sec:hgraphconstr} is a strict subset of all hypergraphs with edges of size three and four. This is because the construction algorithm can be carried out using $\bigO(n^3)$ edge additions, $n$ being the number of marks and thus vertices. However, generic 3,4-hypergraphs can have $\binom{n}{4} \in \bigOmega(n^4)$ edges.

\subsubsection{Forbidden Substructures}

It would be very interesting to determine which sort of hypergraphs can and cannot be constructed; for example through a forbidden subgraph characterization. That is, a set~$F$ of hypergraphs, such that a 3,4-hypergraph~$H$ is a characteristic hypergraph for a ruler if and only if~$H$ does not contain a graph~$G \in F$ in the sense of graph minors, subgraphs or induced subgraphs.  Such a characterization would most likely be very useful in creating efficient algorithms for the problems we discuss in \secref{sec:fpt}. Unfortunately, we were not able to find a forbidden subgraph characterization. However, examples of forbidden and forbidden induced subgraphs have been found, which we discuss now.

\begin{figure}
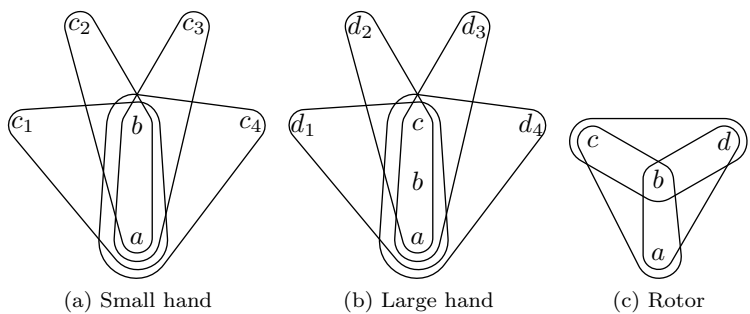

  \begin{center}
    \subfloat[Small hand]{
      \includegraphics{size3-4tooth-fork.1}
      \label{fig:s34tfork}
    }
    \subfloat[Large hand]{
      \includegraphics{size4-4tooth-fork.1}
      \label{fig:s44tfork}
    }
    \subfloat[Rotor]{
      \includegraphics{lollipop.1}
      \label{fig:lollipop}
    }
    
    \caption{Forbidden subgraphs}
    \label{fig:forbiddensub}
  \end{center}
\end{figure}

\begin{lemma}[Small hand forbidden subgraph]
  \label{lem:smallhand}
  The subgraph shown in \autoref{fig:s34tfork} cannot occur in a characteristic hypergraph.
\end{lemma}
\begin{proof}
  In an edge with three marks there is one mark exactly between the other two. Let $a,b$ be two marks on a ruler. Where can a third mark in an edge already comprising $b$ and $a$ be? Either $a$, $b$, or the new mark can be the mark in the middle, and thus there are at most three edges with three vertices intersecting in $a$ and $b$.
\end{proof}

\begin{lemma}[Large hand forbidden subgraph]
  \label{lem:largehand}
  The subgraph shown in \autoref{fig:s44tfork} cannot occur in a characteristic hypergraph.
\end{lemma}
\begin{proof}
  In an edge with four marks there are two pairs of vertices that measure the same distance. We choose one unordered pair from $\{a, b, c\}$ and thus define the distance that caused the edge. Then, for the fourth mark, there are only two possible positions left. Multiplying this with the number of possible unordered pairs, one gets six as an upper bound for such edges intersecting in three marks.

  However, every edge is counted twice here: Assume $a < b$ has been chosen as pair. Then a fourth mark $d$ is defined as follows:
  $$
  \begin{array}{rclcrcl}
    d & = & c - (b - a) & \vee & d & = & c + (b - a)
  \end{array}
  $$
  These are the two possibilities for the fourth mark that are counted above. But we can reformulate these equations as follows:
  $$
  \begin{array}{rclcrcl}
    d & = & a + (c - b) & \vee & d & = & b + (c - a) 
  \end{array}
  $$
  These equations imply that the possibilities that are counted for the pair $a, b$ are identified with one that is counted for the pair $b, c$ and one that is counted for $a, c$. This observation holds for all possible pairs from $\{a, b, c\}$ and thus every possible location for $d$ is counted twice. This means that there are at most three edges comprising four vertices that intersect in three vertices.
\end{proof}

\begin{lemma}[Rotor forbidden subgraph]
  The subgraph shown in \autoref{fig:lollipop} cannot occur in a characteristic hypergraph.
\end{lemma}
\begin{proof}
  Fix a total ordering of the three marks in $\{a, c, d\}$. We then try to position $b$ in that ordering and find that all possibilities lead to a contradiction. Because of the symmetry of the graph we can look at one specific ordering without loss of generality. 

Assume that $a < c < d$. Where can $b$ be put? Assume that $b < c$. Because of the edge $\{b, c, d\}$, $c$ is half-way between $b$ and $d$. The edge $\{a, b, d\}$ implies that either $a = c$ ($\lightning$) or $a < b$. But then, because $a, b$ are in one edge with $c$ and in one with $d$, $c = d$. $\lightning$ The case $b > c$ is symmetric.
\end{proof}

Beside forbidden subgraphs, there are also some forbidden induced subgraphs:

\begin{figure}
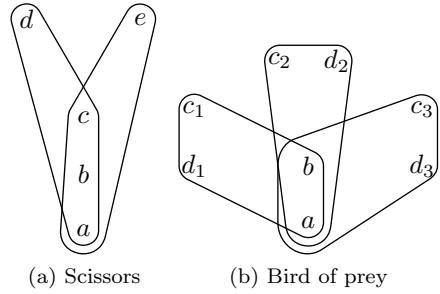

  \begin{center}
    \subfloat[Scissors]{
      \includegraphics{scissors.1}
      \label{fig:scissors}
    }
    \subfloat[Bird of prey]{
      \includegraphics{butterfly.1}
      \label{fig:butterfly}
    }
    
    \caption{Forbidden induced subgraphs}
    \label{fig:forbiddeninducedsub}
  \end{center}
\end{figure}

\begin{lemma}[Scissors forbidden induced subgraph]
  The graph shown in \autoref{fig:scissors} is a forbidden induced subgraph in a characteristic hypergraph.
\end{lemma}
\begin{proof}
  We show that in such a configuration, an edge comprising $d, e$ and exactly one mark from $\{a, b, c\}$ must also be present.
  
  We again use the fact that edges comprising four marks are due to two pairs of them having the same distances. Choose two pairs from $\{a, b, c\}$ corresponding to the two edges. If those pairs comprise the same marks, the statement is trivial. If not, then the pairs must share one mark. Without loss of generality, let the pairs be $a < b$ and $b < c$. The following equations hold:
\begin{align*}
  \begin{split}
    d = c \pm (b - a) 
  \end{split}
  \begin{split}
    e = a \pm (c - b)
  \end{split}
\end{align*}
Note that the sign before $(b - a)$ cannot be negative at the same time with the sign of $(c - b)$ being positive. Otherwise this would imply that $d = e$. In any other case, the two terms on the right hand side of the equations differ only in the sign of exactly two variables. This means that there exists an $m \in \{a, b, c\}$ such that the following equation holds.
\begin{equation*}
  |e - m| = |m - d|
\end{equation*}
Thus there is one edge $\{d, e, m\}$.
\end{proof}

\begin{lemma}[Bird of prey forbidden induced subgraph]
  \label{lem:butterfly}
  The graph shown in \autoref{fig:butterfly} is a forbidden induced subgraph in a characteristic hypergraph.
\end{lemma}
\begin{proof}
  We show that there is a case distinction with two cases for 4-edges that contain two fix vertices~$a,b$ and if there are two 4-edges~$e_i, e_j$ that contain $a, b$ and correspond to the same case, then there is an edge $e_i \cup e_j \setminus \{a, b\}$.

  This implies the statement of the lemma, because if there are three edges that intersect in two vertices then at least one additional edge must exist, making the graph in \autoref{fig:butterfly} a forbidden induced subgraph.
  
  The case distinction is as follows: Fix $a < b$ on a ruler. For the marks~$c_i$ and $d_i$ of an edge~$e_i := \{a, b, c_i, d_i\}$ there are two cases: In the first case, both are between $a$ and $b$ or one to the left of $a$ and one to the right of $b$. In the second case either both marks are positioned left of $a$ or both marks are positioned right of $b$ or one mark is positioned left of $a$ or right of $b$ and the other mark between the two. This distinction is exhaustive.

  The existence of the additional edges can be seen as follows:
  
  For edges $e_i = \{a, b, c_i, d_i\}$ that correspond to the second case, it is clear that the marks in~$e_i \setminus \{a. b\}$ must have the same distance as $a$ and $b$. Thus the existence of extra edges trivially follows in this case.

  Now assume the edges~$\{a, b, c_i, d_i\}$ and $\{a, b, c_j, d_j\}$ are present in a characteristic hypergraph and correspond to the first case. Without loss of generality we assume $c_i < d_i$ and $c_j < d_j$. Then the following equations hold:
\begin{align*}
  \begin{split}
    a - c_j = d_j - b
  \end{split}
  \begin{split}
    a - c_i = d_i - b
  \end{split}
\end{align*}
Subtracting the right from the left one, we get $c_i - c_j = d_j - d_i$. This implies the edge $\{c_i, d_i, c_j, d_j\}$.
\end{proof}

These are some forbidden substructures we have found. There are other more specific sorts of forbidden structures, however, we could not observe a thorough explanation of the structure of the characteristic hypergraphs. We consider some general approaches to this question in the succeeding subsection.

\subsubsection{The Recognition of Characteristic Hypergraphs}

Abstracting from forbidden substructures one can pose the following question:
\probsp
\decprob{\chr}{A 3,4-hypergraph $G$.}{Is there a ruler $R$ such that $H_R \cong G$?}
\probsp
We found two approaches to this question, which could lead to structural insights with proper analysis. We briefly describe them here.

On the one hand, observe that any total order of the vertices in $G$ defines for every edge or rather conflict a linear equation in the vertices or marks that are present in it. That is, any total order of vertices of $G$ defines a homogeneous linear equation system with the vertices as variables, a solution to which can provide the sought ruler. However, certain solutions are infeasible: the trivial solution identifying every mark with zero, any other solution that requires two marks to be equal, and any solution that introduces edges that are not present in the graph. One challenge is to classify equation systems that exclude the second and third kind of solutions, another challenge is to either eliminate the need for a total ordering of the vertices or to find a way of deriving a subclass of feasible orderings in order for this classification to be efficient. A third challenge arises from the possibility of ``derived edges'', that is, edges whose linear equations are linear combinations of other edges' equations (as we observed in \autoref{lem:butterfly}).

\begin{figure}
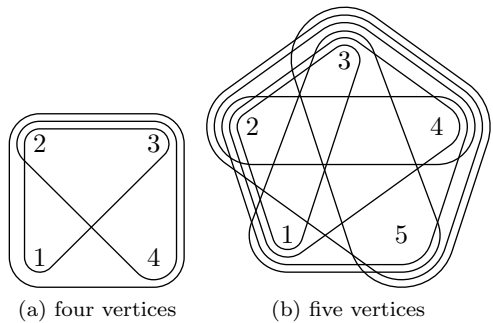

  \begin{center}
    \subfloat[four vertices]{
      \includegraphics{complete4.1}
      \label{fig:complete4}
    }
    \subfloat[five vertices]{
      \includegraphics{complete5.1}
      \label{fig:complete5}
    }
    
    \caption{Here, characteristic hypergraphs of the rulers $\{1, 2, 3, 4\}$ and $\{1, 2, 3, 4, 5\}$ are shown. Since such graphs have the asymptotically maximum number of edges, we call them pseudo-complete. Any ruler of length $\leq 5$ has a characteristic hypergraph that is a subgraph of one of the graphs shown. Perhaps this property can be used to efficiently recognize characteristic hypergraphs.}
    \label{fig:complete}
  \end{center}
\end{figure}

Another approach to \chr{} would be the following: Because every ruler $R$ is a subset of a ruler containing every integer between the maximum and minimum of $R$, one can investigate the question, whether there is an integer $n$ such that $G$ is isomorphic to an induced subgraph of $H_{\{1, 2, ..., n\}}$. (See examples of such pseudo-complete graphs in \autoref{fig:complete}.) This also would give a hint on whether $G$ has a constructing ruler. 

An NP-hardness result for \chr{} of course would be a pessimal alternative. This would also mean that there is no classification through finitely many, constant-size, forbidden (induced) subgraphs. 

\subsubsection{Other Properties}

Planarity and bounded degree are other properties that could possibly be exploited for efficient algorithms. However, in general the characteristic hypergraphs are neither:

\begin{figure}%[14]{r}{0.37\textwidth}
  \begin{center}
    \includegraphics{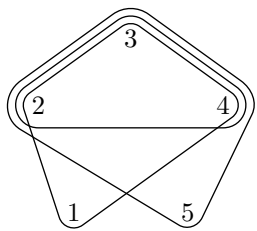}
    \caption{A subgraph of $H_{\{1, 2, ..., 5\}}$ with three edges is shown. The incidence graph of this graph contains a $K_{3,3}$ and thus is nonplanar.}
    \label{fig:nonplanar}
  \end{center}
\end{figure}

As we observed in \autoref{sec:hgraphconstr}, there are characteristic hypergraphs with $n$ vertices and $\bigOmega(n^3)$ edges. However, if a graph has bounded degree, i.e., there are at most $c$ edges adjacent to any vertex, it has at most $c n$ edges for a constant~$c$.

A hypergraph is called planar if and only if its incidence graph representation is planar. We show an example of a hypergraph whose incidence graph comprises a complete bipartite graph with two vertex partitions of size three---a $K_{3,3}$---as a minor. Look at the three-edged subgraph of $H_{\{1, 2, ..., 5\}}$ shown in \autoref{fig:nonplanar}. All three edges intersect in 2, 3, and 4, and thus this graph's incidence graph has a $K_{3,3}$ as a subgraph. Obviously, this graph is then also part of the incidence graph of $H_{\{1, 2, ..., 5\}}$.

\section{A Simplified NP-Hardness Proof for \grsm}
\label{sec:simplnphardness}

\citet{MP08} showed that the following problem is NP-hard via a reduction from a hard SAT-variant.
\probsp{}
\decprob{\grsm}{A finite set $R \subseteq \mathbb{N}$ and $l\in\mathbb{N}$.}{Is there a \gr{} $R' \subseteq R$ with at least $l$ marks?}
\probsp{}
However, the construction of the ruler $R$ corresponding to the SAT-formula is involved and hard to comprehend. We provide a reduction from a different SAT-variant, yielding a much simpler construction. In a simple corollary, we also gather a W[1]-hardness result for a similar problem, that did not directly follow from the original proof. To define the problem, which we reduce from, we first need the following definition:
\begin{definition}
  A boolean formula in 2-CNF is said to be \emph{antimonotone} if and only if every literal is negative.
\end{definition}
%\probsp{}
\decprob{\wacs}{An antimonotone boolean formula $\phi$ in 2-CNF and an integer~$k$.}{Is there a truth assignment for the variables in $\phi$ such that $\phi$ is satisfied and at least $k$ variables are assigned true?}
%\probsp{}
\paragraph{NP-Completeness.} \label{anc:istowacs} This problem is NP-complete via a straightforward reduction from \textsc{Independent Set}. In \textsc{Independent Set}, a graph $G = (V,E)$ and an integer $k$ is given and it is asked whether there is a vertex set $S \subseteq V$ of size at least $k$ such that there is no $e \in E$ with $e \subseteq S$. The set $S$ is called an independent set of~$G$.

Let $G=(V,E)$ be a graph. It is clear, that in any edge $\{u, v\} \in E$ either~$u$ or $v$ is not in an independent set of $G$. The following formula is true for any independent set $S$ of $G$, where $\chi_S$ is the characteristic function of the set~$S$. $$\phi(G) := \bigwedge_{\{u, v\} \in E} (\neg \chi_S(u) \vee \neg \chi_S(v))$$ If we interpret the instances of $\chi$ as variables and $\phi(G)$ has a satisfying truth assignment with at least~$k$ positive variables, there is an independent set of size~$k$ for~$G$ and vice versa.

\paragraph{Problem Restriction.} We restrict the problem \wacs{} to instances such that $\phi$ contains no two clauses consisting of the same variables. It is a simple observation that this has no impact on the complexity. (Note that such clauses cannot occur in the formula $\phi(G)$ in the construction above.)

\paragraph{Reduction Outline.} The basic idea of our reduction is to construct a ruler such that there are marks for every variable and additional marks that impose conflicts between those and only those marks that correspond to variables that share a clause. In terms of characteristic hypergraphs, we construct a ruler $R$, such that $H_R$ contains only edges that consist of two marks that correspond to two variables and two auxiliary marks. Every edge shall correspond to one clause.

The construction is given below. It is followed by some auxiliary statements that we use to show that the conflicts corresponding to the clauses are the only conflicts present in the constructed ruler. Once we have done that, we are able to prove that \wacs{} can be solved with \grsm.

%TODO redefine paragraph to conform with theorems

\begin{construction}
  \label{con:red}
  Let $\phi$ be an antimonotone boolean formula in 2-CNF with $m$ clauses (numbered from $1$ to $m$), let $x_1, ..., x_n$ be the variables in $\phi$ and $k$ an integer. Construct a ruler~$R$ and an integer~$l$ as follows:

  Let $V = \{ 2^{(m + 2) i} : x_i \text{ in } \phi\}$ be the set of marks corresponding to the variables and let $C_i = \{ 2^{(m + 2)i} + 2^j - 1 : x_i \text{ occurs in the } j \text{'th clause of } \phi\}$ and $C = \bigcup_i C_i$ be the sets of marks corresponding to the clauses. 

The ruler $R = V \cup C$ and $l = k + 2m$ then constitute an instance of \grsm .
\end{construction}
Before we prove the polynomial running time and correctness, i.e., the equivalence of instances of this construction, we need some auxiliary statements:
\begin{lemma}
  \label{lem:2^natgr}
  The ruler $\twonat := \{2^i : i \in \mathbb{N}\}$ is a \gr.
\end{lemma}
\begin{proof}
  If $\twonat$ is not a \gr, then there are integers $i < j \leq k < l$ such that the following holds:
  \begin{align*}
    2^j - 2^i &= 2^l - 2^k\\
    2^{j - i} - 1 &= 2^{l - i} - 2^{k - i} \\
    2^{j - i} + 2^{k - i} - 1 &= 2^{l - i} 
  \end{align*}
  However, because $j,k < l$, the right-hand side of this equation is strictly greater than the left-hand side. $\lightning$
\end{proof}
\begin{corollary}
  \label{cor:vandcigr}
  The ruler $V$ and each $C_i$ is a \gr .
\end{corollary}
\begin{proof}
  The ruler $V$ is a subset of $\twonat$, and each $C_i$ is a subruler of the following ruler: $$\{ 2^{(m+2)i} - 1 + j : j \in \twonat \} \qedhere$$
\end{proof}

Now we know that the basic building blocks of the constructed ruler are conflict-free. We proceed to show that $C$ also is conflict-free. Because conflicts do contain at most four marks, it suffices to prove that $C_i \cup C_j \cup C_k \cup C_l$ contains no conflicts for any choice of $i, j, k, l$. We do that by successively adding more clause sets $C_o$. 

While doing this we maintain that any one mark in $V$ can be added to the unions of clause sets without imposing conflicts. Together with the observation that any one mark in a clause set can be added to $V$ without creating a conflict, we then see that there are no conflicts in $R$ that contain one mark from $V$ and three of $C$ or vice versa.

To continue, we need the following auxiliary results that are a modification of \autoref{lem:2^natgr}.

\begin{lemma}
  \label{lem:2^natapproxgr}
  For any integer $d \geq 2$ and any function~$f: \mathbb{N} \rightarrow \{0, ..., 2^{d} \}$, the following ruler is a \gr{}: 
  $$\twonatapprox(d,f) := \{2^{di} + f(i) : 1\leq i \in \mathbb{N} \}$$
\end{lemma}
\begin{proof}
  If there is such a ruler that is not a \gr{}, then there are four marks in it---i.e. $i < j \leq k < l$ and constants~$0 \leq m_i, m_j, m_k, m_l \leq 2^{d}$---such that the following holds:
  \begin{align*}
    (2^{dj} + m_j) - (2^{di} + m_i) &= (2^{dl} + m_l) - (2^{dk} + m_k)
  \end{align*}
  If we divide this equation by $2^{di}$, we get the following:
  \begin{align*}
    (2^{d(j-i)} + \hat{m_j}) - (1 + \hat{m_i}) &= (2^{d(l-i)} + \hat{m_l}) - (2^{d(k-i)} + \hat{m_k})\\
    (2^{d(j-i)} + \hat{m_j}) + (2^{d(k-i)} + \hat{m_k}) - (1 + \hat{m_i}) &= (2^{d(l-i)} + \hat{m_l})
  \end{align*}
  Since $0 \leq \hat{m_i}, \hat{m_j}, \hat{m_k}, \hat{m_l} \leq 1$ and $j, k \leq l - 1$, the left-hand side of the above equation is upper bounded by $2^{d(l-i) - d + 1} + 1$. However, the right-hand side is lower bounded by $2^{d(l-i)}$. This is a contradiction, because the following relations hold: $1 \leq l - i$ and~$2 \leq d$. $\lightning$
\end{proof}
\begin{observation}
  \label{obs:mindist}
  The distance between two marks in $\twonatapprox(d,f)$ is at least~$2^{2d - 1}$.
\end{observation}
\begin{proof}
  Let $2^{dj} + m_j$ and $2^{di} + m_i$ with $i < j$ be two marks in $\twonatapprox(d,f)$. Then the following relations hold:
  $$
  \begin{array}{clclcl}
        & 2^{dj} + m_j - (2^{di} + m_i)    &\geq & 2^{dj} - (2^{di} + 2^d)     & \geq &  2^{d(i + 1)} - 2^{di} - 2^d \\
      = & 2^{di}(2^d - 1) - 2^{d}          &\geq & 2^{d}(2^d - 1) - 2^{d}      & =    &   2^{2d} - 2^{d + 1}          \\
      = & 2 \times 2^{2d - 1} - 2^{d + 1}  &\geq & 2^{2d - 1} & & 
  \end{array}
  $$
  The last estimate holds because $d \geq 2$.
\end{proof}
\begin{lemma}
  \label{lem:2^natapproxaddgr}
  The ruler $\twonatapprox(d,f)$ does not cease to be a \gr{}, if we add the mark~$2^{di} + c$ to it, where $i, c \in \mathbb{N}$ and $0 \leq c \leq 2^d$.
\end{lemma}
\begin{proof}
  It is clear that any edge $e$ in the characteristic hypergraph of the ruler~$\twonatapprox(d,f) \cup \{2^{di} + c \}$ must comprise the marks~$\alpha := 2^{di} + f(i)$ and $\beta := 2^{di} + c$. The distance measured by those two marks is at most $2^d$ and the distance measured by any two marks in $\twonatapprox(d,f)$ is at least $2^{2d-1}$ (\autoref{obs:mindist}). Since $d \geq 2$, this means that the edge $e$ is not due to the distance~$|\alpha - \beta|$. Because of \autoref{lem:nonoverlap} this also means that we can assume the marks~$\alpha$ and $\beta$ to be positioned between the other two marks in~$e$. The existence of such an edge, however, has been disproven in \autoref{lem:2^natapproxgr}: Have a look at the proof and observe, that we did not postulate the inequality of $j$ and~$k$ and also did not make an assumption about the constants $m_j$ and~$m_k$.
\end{proof}
\begin{observation}
  \label{obs:maxmindist}
  Let $v = 2^{(m +2)i} \in V$, $c_i \in C_i$ and $c_j, \hat{c_j} \in C_j$ with $i \neq j$. The following relations hold:
  \begin{align*}
    \begin{split}
      |v - c_i|          & \leq 2^m \\
      |c_j - \hat{c_j}|   & \leq 2^m
    \end{split}
    \begin{split}
      2^{2m + 3}  & \leq |v - c_j| \\
      2^{2m + 3}  & \leq |c_i - c_j| 
    \end{split}
  \end{align*}
\end{observation}
\begin{proof}
  The relations on the left are trivial. Those on the right-hand side follow from \autoref{obs:mindist} by modeling $\{v, c_j\}$ and $\{c_i, c_j\}$ as subrulers of the ruler~$\twonatapprox(m + 2, f)$ with an appropriate function~$f$.
\end{proof}
\begin{corollary}
  \label{cor:vccjgr}
  The rulers $V^c := V \cup \{c\}, c \in C$ and $C_j^v := C_j \cup \{v\}, v \in V$ are \gr s for any integer $j$, $1 \leq j \leq n$.
\end{corollary}
\begin{proof}
  The ruler $V^c$ can be modeled as subruler of $\twonatapprox(m + 2, f) \cup \{c\}$ for appropriate values of $f$. According to \autoref{lem:2^natapproxaddgr}, this is a \gr .

  Let $v = 2^{(m + 2)i}$ and assume that $C_j^v$ contains a conflict. It is not possible that~$j = i$, because otherwise the ruler $C_j \cup \{v\}$ is a shifted subset of the \gr{}~$\twonat$. The case~$j \neq i$ is also not possible: A distance measured by~$c_j, \hat{c_j} \in C_j$ is at most~$2^m$, however, it is $2^{2m + 3} \leq |v - c_j|$ because of \autoref{obs:maxmindist}. $\lightning$
\end{proof}
\begin{lemma}
  \label{lem:cicjgr}
  The ruler $C_i \cup C_j$ is a \gr{} for all integers $i,j$ with $1 \leq i < j \leq n$.
\end{lemma}
\begin{proof}
  Assume that there is an edge $e$ in $H_{C_i \cup C_j}$. This edge is due to two equal distances. A distance here can be measured by two marks in $C_i$ (a), by two marks in $C_j$ (b), or by one mark in $C_j$ and one mark in $C_i$ (c). 

  Because of \autoref{cor:vandcigr}, the edge $e$ cannot be due to two distances both corresponding to (a), or (b). Because of \autoref{lem:nonoverlap} the case that both distances correspond to (c) reduces to the case that one distance corresponds to (a) and one to (b).

  Assume that the edge $e$ is due to two distances corresponding to (a) and (b), respectively. Then there are $k < l$ and $o < p$ such that $2^l - 2^k = 2^p - 2^o$ holds. We observed in \autoref{lem:2^natgr} that $\twonat$ is a \gr . That means that~$l = p$ and $k = o$. Then, the variables $x_i$ and $x_j$ occur together in clause $l$ and $k$. However, we excluded such instances from our considerations.

  Assume that the edge $e$ is due to one distance $d_1$ corresponding to (a) and one distance $d_2$ corresponding to (c). Because of \autoref{obs:maxmindist}, the following relation holds.
  \begin{align*}
    d_1 \leq 2^m < 2^{2m + 3} \leq d_2
  \end{align*}
  Thus, $d_1$ and $d_2$ cannot be equal. $\lightning$
  
  The case (b) and (c) is analog to (a) and (c).
\end{proof}
\begin{corollary}
  The ruler $(C_i \cup C_j)^v := C_j \cup C_i \cup\{v\}$ is a \gr{} for all marks~$v \in V$ and integers~$i, j$ with $1 \leq i < j \leq n$.
\end{corollary}
This follows directly from \autoref{cor:vccjgr} and the proof of \autoref{lem:cicjgr}, because every statement there still holds if we redefine $C_i$ to contain $v$.
\begin{lemma}
  \label{cor:cicjckgr}
  The ruler $C_i \cup C_j \cup C_k$ is a \gr{} for all integers~$i,j,k$ with $1 \leq i < j < k \leq n$.
\end{lemma}
\begin{proof}
  Assume that there is an edge $e$ in $H_{C_i \cup C_j \cup C_k}$. Because of \autoref{lem:cicjgr}, $e$ must comprise marks from every of the three clause-sets. 

  Assume that $e$ is a 3-edge. It consists of one mark in $C_i$, one in $C_j$ and one in $C_k$. However, these three marks form a ruler that is a subruler of $\twonatapprox(d,f)$ with appropriate values for $d$ and $f$.~$\lightning$
  
  Now assume that $e$ is a 4-edge. Because there are three clause-sets and $e$ must comprise marks from every of them, $e$ contains two marks $m_1$ and $m_2$ from exactly one of the sets. Then, however, $e$ forms a subset of a ruler corresponding to \autoref{lem:2^natapproxaddgr}.~$\lightning$
\end{proof}
\begin{corollary}
   The ruler $C_i \cup C_j \cup C_k \cup \{v\}$ with $v \in V$ is a \gr{} for all integers $i,j,k,l$ with $1 \leq i < j < k \leq n$.
\end{corollary}
This again directly follows from the proof of \autoref{cor:cicjckgr}.
\begin{lemma}
  \label{lem:cgr}
  The ruler $C$ is a \gr.
\end{lemma}
\begin{proof}
  Because of \autoref{lem:cicjgr} and \autoref{cor:cicjckgr} it only remains to show that $C_i \cup C_j \cup C_k \cup C_l$ is a \gr . This however is not hard to see, as any edge in the corresponding characteristic hypergraph must consist of four marks from every one of the clause-sets. Such an edge again would be a subruler of~$\twonatapprox(m + 2,f)$ with an appropriate function~$f$.~$\lightning$
\end{proof}
\begin{corollary}
  The ruler $C \cup \{v\}, v \in V$ is a \gr .
\end{corollary}
\begin{corollary}
  \label{cor:2-2-edges}
  If there are edges in $H_R$, there only are 4-edges that contain two marks from $V$ and two from $C$.
\end{corollary}
\begin{lemma}
  \label{lem:clausesedges}
  If there is an edge $e$ in $H_R$, then it intersects with $V$ in two marks corresponding to the variables $x_i$ and $x_j$ and it intersects in exactly one mark with $C_i$ and in exactly one mark with $C_j$. The variables $x_i$ and $x_j$ are together in one clause in $\phi$.
\end{lemma}
\begin{proof}
  Let $e$ be an edge in $H_R$. Because of \autoref{cor:2-2-edges} it intersects with $V$ in two marks. Let these two marks be $v_i = 2^{(m + 2) i}$ and $v_j = 2^{(m + 2) j}$, $i < j$, i.e. the marks corresponding to the variables $x_i$ and $x_j$ in $\phi$. 

  The edge $e$ cannot contain a mark $c_k \in C_k$ with $k \notin \{i,j\}$, because then $e$ would be a subset of a ruler corresponding to \autoref{lem:2^natapproxaddgr}.  However, $e$ can also not contain two marks $c_k, \hat{c_k} \in C_k$ with $k \in \{i, j\}$: Any distance measured by marks in~$\{v_k, c_k, \hat{c_k}\}$ is at most $2^m$ and those three marks are at least $2^{2m+3}$ units away from the fourth mark (\autoref{obs:maxmindist}).

  Let $e$ contain the mark $c_i = 2^{(m + 2) i} + 2^k - 1 \in C_i$ and the fourth mark~$c_j \in C_j$. Then, the following equations hold.
  \begin{align*}
    c_i - v_i &= c_j - v_j &  & \Leftrightarrow & c_j &= 2^{(m + 2) j} + 2^k - 1
  \end{align*}
  This means that $x_i$ and $x_j$ are together in clause $k$.
\end{proof}
\begin{lemma}
  \label{lem:clauseedgeinter}
  Let $e,f$ be two different edges in $H_R$, then $e \cap f \subseteq V$.
\end{lemma}
\begin{proof}
  Assume that $e,f$ have a non-empty intersection and $e \cap f \cap C_k \neq \emptyset$. As we observed in \autoref{lem:clausesedges}, $e \cap f$ also contains $v_k \in V$. This means that either $e = f$ or there are three variables in one clause.~$\lightning$
\end{proof}
We are now ready to prove the following:
\begin{lemma}
  \wacs{} is polynomial-time many-to-one reducible to \grsm.
\end{lemma}
\begin{proof}
   Let $\phi$ be an antimonotone boolean formula in 2-CNF with $m$ clauses, let $x_1, ..., x_n$ be the variables in $\phi$ and let $k$ be an integer, that is, let $\phi$ and~$k$ constitute an instance of \wacs. Construct an instance for \grsm{}, i.e., a ruler~$R = V \cup C$ and an integer~$k + 2m$ according to \autoref{con:red}. 

  The marks to be constructed are given explicitly by the variables and their position in the clauses and the number of marks is polynomial in the number of variables and clauses. The length of the binary encoding of the marks is also bounded by a polynomial in the number of variables and clauses and thus the construction is possible in polynomial time.

  Given a truth assignment that satisfies~$\phi$ with $k$ positive variables, we can find a \gr{} with $k + 2m$~marks in $R$: We have seen in \autoref{lem:clausesedges} that every edge in $H_R$ corresponds to a clause. In one clause at most one variable can be assigned true. Therefore the marks corresponding to the positive variables form a \gr{} with $k$~marks. Additionally, for every edge~$e$ in $H_R$, $e \cap C$ contains only degree-one vertices, because of \autoref{lem:clauseedgeinter}. Because every edge contains at least one mark corresponding to a negative variable, we can add these degree-one vertices to the \gr{} constructed from the positive variables, yielding a \gr{} with $k + 2m$~marks.

  Given a \gr{} $R' \subseteq R$ with $k + 2m$~marks, we construct a truth assignment that satisfies~$\phi$ with $k$~positive variables: We assign every variable that corresponds to a mark in $R'$ the value true, every other variable is assigned false. Because there are exactly $n + 2m$ marks in $R$, any such truth assignment has at least $k$ positive variables. However, this simple assignment does not necessarily satisfy $\phi$: There may be two marks in $V$ and $R'$ that are in the same clause in $\phi$ and thus are in the same edge in $H_R$. But in every such edge $e$, there must be one mark $r_e$ that is not in $R'$ and, thus, we can simply exchange an arbitrary mark from $e \cap V$ in $R'$ with $r_e$ and get an assignment that satisfies~$\phi$.
\end{proof}
Now the following theorem directly follows:
\begin{theorem}
  \grsm{} is NP-complete, even if there are only conflicts with four marks in the input instance.
\end{theorem}
The reduction also yields a W[1]-hardness result for a modified problem:
\probsp{}
\decprob{Golomb Subruler $\geq$ Double Conflicts}{A ruler $R \subseteq \mathbb{N}$ and $k\in\mathbb{N}$.}{Is there a \gr{} $R' \subseteq R$, such that $|R'|$ is at least $k$ plus two times the number of edges in $H_R$?}
\probsp{}
% TODO? W[1]-hardness for special parameter Consider the following modified problem:
\begin{corollary}
  \label{anch:grsm'w-hard}
  \textsc{Golomb Subruler $\geq$ Double Conflicts} is W[1]-hard with respect to parameter $k$.
\end{corollary}
\begin{proof}
  Observe that there is a parameterized reduction from \textsc{Independent Set} parameterized with the sought independent set size to \wacs{} parameterized with the number of positive variables in a satisfying truth assignment (we have given the idea for this on page \pageref{anc:istowacs}). The reduction from \wacs{} to \grsm{} we have given above maps a formula $\phi$ and the parameter ``number of positive variables'' $k$ to a ruler $R$ and the sought ruler size $k + 2m$, where $m$ is the number of edges in $H_R$. This means, this reduction identifies the parameters of \wacs{} and \textsc{Golomb Subruler $\geq$ Double Conflicts}, making it a parameterized reduction.
\end{proof}

\section{Fixed-Parameter Tractability of Constructing Golomb Rulers}
\label{sec:fpt}

The number of implementations of the search for optimal or near-optimal \gr s hints to the importance of this problem (see e.g. \cite{CF05, CDFH07, Dim02, Distrib, PTC03, Ran93, TPC07}). Unfortunately, several problems closely related to the construction of \gr s have been proven to be NP-complete (cf. \secref{sec:prcomplex}).

We now look at the fixed-parameter tractability of constructing \gr s with the goal to get exact and relatively efficient algorithms. In particular, we focus on the following problem.
\probsp{}
\decprob{\grsm}{A finite set $R \subseteq \mathbb{N}$ and $n\in\mathbb{N}$.}{Does there exist a \gr{} $R' \subseteq R$ of at least $n$ marks?}
\probsp{}
For a given $R$ in \grsm, let $H_R=(R,E)$ be the characteristic hypergraph as defined in \secref{sec:hgraphconstr}. Now the task is to find a subset $R' \subseteq R$ such that the induced subgraph $H_{R'}=H_R[R']$ contains no edges. This can either be done by selecting marks to keep or by deleting marks from $R$. 

\subsection{Mark Deletion Parameter}
\label{sec:markdeletion}

If we decide to delete marks, we can parameterize \grsm{} with the maximum number of allowed mark deletions. Together with the notion of hypergraph characterization, we can reformulate it as follows.
\probsp{}
\decprob{\grmd}{A finite set $R \subseteq \mathbb{N}$ and $k\in\mathbb{N}$.}{Is there a set of marks $D$ with size at most $k$, such that $H_{R \setminus D}$ contains no edges?}
\probsp{}
The above described graph problem is strikingly similar to the \textsc{Hitting Set} problem. In \textsc{Hitting Set}, a hypergraph is given and a (minimum size) subset~$S$ of vertices is sought, such that every edge has at least one vertex in $S$. In fact, our problem can canonically (by simply computing the characteristic hypergraph) be reduced to this problem parameterized with the size of $S$ in polynomial time. 

However, as we observed in \autoref{sec:hgraphstruc}, our graph instances are a strict subset of all possible hypergraph instances, which raises hope for better algorithms than those known for this generic problem.
\subsubsection{Fixed-Parameter Algorithm}
\label{sec:mdalgorithm}

From the hypergraph characterization and the notion of deleting vertices we immediately get a search tree algorithm: It is clear that in every edge of a characteristic hypergraph at least one mark has to be deleted. That means, one can choose one edge, branch into all possibilities of deleting one mark in that edge and do this recursively until either $k$ marks have been deleted or the characteristic graph has no edges. (See also the pseudocode in \autoref{alg:searchtree}.)

\begin{figure}
  \begin{algorithm}[H]
    \KwIn{A hypergraph $H = (R, E)$, an integer $k$ and a vertex deletion set~$D$.}
    \KwOut{A set of vertices $D$ of size $\leq k$ such that $H[R \setminus D]$ contains no edges, if it exists.}

    \SetKwFunction{STN}{SearchTreeNode}

    \lIf{$H[R \setminus D]$ has no edges}{
      Output $D$ and halt\;}
    \lIf{$|T| = k$}{
      Abort this branch\;}
    \Else{
      Choose an edge $e \in E$\;
      \lFor{$i \in e$}{
        \STN{$H, k - 1, D \cup \{i\}$}}}
    
    \SetAlgoRefName{SearchTreeNode}
    \caption{Solving \grmd .}
    \label{alg:searchtreenode}
  \end{algorithm}
  \begin{algorithm}[H]
    \KwIn{A ruler $R \subset \mathbb{N}$ and an integer $k$.}
    \KwOut{A set of marks $D$ of size $\leq k$ such that $R \setminus D$ is Golomb, if it exists.}
    
    \SetKwFunction{STN}{SearchTreeNode}
    \SetKwFunction{HGC}{HypergraphConstruction}

    $H_R \leftarrow$ \HGC{R}\;
    \STN{$H_R,k,\emptyset$}

    \SetAlgoRefName{Search\-Tree}
    \caption{Solving \grmd .}
    \label{alg:searchtree}
  \end{algorithm}
\end{figure}

Because every edge in a graph defined by a ruler $R$ as in \secref{sec:hgraphconstr} has edges with at most four vertices, and the recursion depth of \ref{alg:searchtreenode} is at most $k$, the running time of this algorithm is time-bounded by a term in $\bigO^*(4^k)$. That is, this algorithm suffices to classify this problem as fixed-parameter tractable with respect to at most $k$ mark deletions.

As we noted above, \grmd{} can be solved with \textsc{Hitting Set} algorithms. The fastest known algorithm for \textsc{Hitting Set} with edges of four vertices and parameterized with the solution size $k$ is $\bigO(3.076^k + m)$ by \citet{DGHNT10}. % TODO? mainly using a result by wahlstroem...
For this more restricted problem it is probably possible to beat this upper bound. Unfortunately we did not find such an algorithm. However, there are some applicable heuristic tricks that improve the above mentioned search strategy in practice:

\paragraph{Edges with Three Vertices.} In every edge with three vertices, at least one vertex has to be deleted, to make the graph edge-free. Thus, in \ref{alg:searchtreenode} one can search for such an edge, and branch into the deletion of every one vertex in it. This leads to a branching vector of $(1, 1, 1)$ and the search for the edge takes $\bigO(m)$ time, $m$ being the number of edges in the graph. However, there are characteristic hypergraphs that do not comprise edges with three vertices and thus this rule does not suffice to improve the theoretical upper bound for the search running time.

\paragraph{Dominating Vertices.} Another simple rule for an improved branching vector is the search for dominating vertices. 
\begin{definition}[Dominating vertices]
  \label{def:domvertices}
  A vertex $v$ dominates a vertex $u$, if and only if $v$ is in every edge that $u$ is incident to. 
\end{definition}
It is clear, that if there is an optimal solution that contains a vertex $u$ that is dominated by another vertex $v$, there is also an optimal solution that contains $v$ instead of $u$ (and obviously an optimal solution does never comprise both). Thus, in \ref{alg:searchtreenode} one can search for a vertex $v$ that dominates another vertex, choose an edge that contains both and branch on either deleting $v$ or any one of the vertices in this edge, that are not dominated by $v$. The worst case for this rule is, that there are only vertices that dominate exactly one other vertex, leading to a branching vector of $(1, 1, 1)$. The search for dominating vertices can be conducted in time $\bigO(nm^2)$ by iterating over every vertex $v$ and checking for every adjacent vertex, whether its incident edge set is a subset of the edge set of $v$.

Also notice that this rule doubles as a reduction rule for edges that intersect in at most one fix vertex $v$ with any other edge in the graph.

\paragraph{Cementating Branched-on Vertices.} \label{anch:cementating} A strategy that is applicable to any deletion search is the ``cementating'' of an entity $v$, when the recursive call for the deleted $v$ returns and no solution has been found. It is then clear, that no solution with $v$ deleted is possible in further branching and $v$ can be prohibited from deletion or be ``cementated''.

This strategy however can be extraordinarily powerful in the search routine for our particular problem: Suppose that at a call of \ref{alg:searchtreenode} a set of vertices $C$ has been cementated. This means, that if there is a solution to be found through further branching, this solution does not contain any vertex in $C$, and $C$ is a \gr. This can already be exploited for an early termination of the branch, if $C$ is not. Furthermore, it is clear that any distance between vertices in $C$ must not appear in the rest of the vertices. This has two implications: First, if a distance appearing between vertices in $C$ is also measured by one vertex $v$ in $C$ and one vertex $u$ not in $C$, then $u$ has to be added to the solution. Second, any such distance measured by two vertices not in $C$ has to be destroyed, leading to a branching vector of $(1, 1)$. 

This rule is very powerful in practice as we observe in \secref{sec:impl} and with further analysis might also lead to an improvement of the theoretical upper bound on the running time of the search.

\subsubsection{Cubic Problem Kernel}
\label{sec:mdkernel}

\citet{Abu09} observed, that \textsc{Hitting Set}, parameterized with the maximum number of vertices $d$ in an edge and the maximum solution size $k$, has a problem kernel of at most $\bigO(k^{d-1})$ edges. However, the reduction rules are not directly applicable to our problem. This is because instances produced by those reduction rules may not correspond to a ruler anymore: Some rules shrink edges, but since the edges correspond to conflicts in characteristic hypergraphs, there is no equivalent to a shrunk edge.

However, there are adequate substitute rules for our hypergraph instances that also retain the problem kernel size of $\bigO(k^3)$ edges (in characteristic hypergraphs $d=4$). These substitute rules compared to the ones by \citet{Abu09}, we salvage the basic idea of the high occurrence rules and use some structure we observed in characteristic hypergraphs to make them applicable.

In the following we assume that the characteristic hypergraph of the input ruler has been computed, and is kept updated alongside the ruler. We first list some simple and obvious rules. Then we give a rule that suffices to bound the number of 3-edges in the characteristic hypergraph. With an additional observation we give another rule to bound the number of 4-edges in the reduced graph. With the help of these two bounds, we are then able to bound the number of marks in a reduced instance.

\begin{rrule}[Lone edges]
  \label{rrule:loneedges}
  If there is an edge that does not intersect with any other edge, then remove it and all vertices it comprises from the graph and reduce $k$ by one.
\end{rrule}
\begin{rrule}[Lone vertices]
  \label{rrule:lonevertices}
  If there is a vertex with degree zero, then remove it from the graph.
\end{rrule}
\begin{rrule}[Leaf edges]
  \label{rrule:leafedges}
  If there is an edge that intersects any other edge in at most one of its vertices $v$, then remove $v$, remove any edges incident to $v$ and decrement $k$.
\end{rrule}
It is clear, that these three simple rules are correct and can be executed in time $\bigO(n+m)$.
\begin{rrule}[High Degree for 3-Edges]
  \label{rrule:highdegree3}
  If there is a vertex that has more than $3k$ incident 3-edges, then remove it from the graph, remove any incident edges and reduce $k$ by one.
\end{rrule}
\begin{lemma}
  \label{lem:highdegree3}
  \autoref{rrule:highdegree3} is correct and can be carried out in running time $\bigO(n + m)$. A graph has at most $3k^2$ 3-edges, if this rule cannot be applied to it and it is solvable with $k$ vertex deletions.
\end{lemma}
\begin{proof}
  Assume there are more than $3k$ 3-edges incident to one vertex $v$. We have seen in \autoref{sec:hgraphstruc}, \autoref{lem:smallhand} that there are at most three edges with three vertices that intersect in two vertices. That means the deletion of any other vertex in the graph can destroy at most three edges incident to $v$. Thus, if $v$ is not deleted, at least $k + 1$ vertices are necessary to destroy all edges incident to $v$.
  
  To apply \autoref{rrule:highdegree3}, one can simply iterate over every 3-edge and count the occurrence of the vertices in an array indexed by the vertices. This is possible in time $\bigO(n + m)$.

  Now assume \autoref{rrule:highdegree3} cannot be applied to the yes-instance $H$. Every edge in $H$ has to be destroyed, 3-edges in particular. One vertex can hit at most $3k$ 3-edges and thus, $H$ has at most $3k^2$ 3-edges.
\end{proof}

\begin{lemma}[Induced Clique Structures]
  \label{lem:inducedclique}
  If there are more than $3k$ 4-edges that intersect in two vertices, this instance cannot be solved with at most $k$ vertex deletions.
\end{lemma}
% TODO: |a - b| = | c - d | or not.
\begin{proof}
  Recall that if there are more than two 4-edges intersecting in two vertices, there are additional edges not containing those vertices, as we observed in \autoref{lem:butterfly} in \autoref{sec:hgraphstruc}. In fact, in this lemma we proved a stronger result: There is a case distinction with two cases for 4-edges that intersect in two vertices $a, b$ and if the edges $e_1$ and $e_2$ correspond to the same case, there is an edge $(e_1 \cup e_2) \setminus \{a, b\}$.

  In case one, the marks other than $a, b$ are in between $a$ and $b$ or one to the left of $a$ and one to the right of $b$. That means, edges that correspond to case one cannot intersect in vertices other than $a, b$ and if there are more than $k - 1$ such edges, there is a structure with more than $k$ pairs of vertices that are all pairwise in an edge. It is clear that in such a structure, more than $k$ vertices have to be deleted, to make it edge-free.
  
  In case two, the two marks other than $a, b$ are positioned left of $a$ or both marks are positioned right of $b$ or one mark is positioned left of $a$ or right of $b$ and the other mark between the two. In this case $a, b$ measure the same distance as the other two marks in the edge. That means, at most two edges in this case can intersect in a vertex other than $a$ and $b$. Thus, if there are more than $2k$ edges that intersect in two vertices and correspond to case two, then there is a structure with more than $2k + 1$ pairs of vertices that are all pairwise in an edge. At most two of those pairs overlap in one vertex, and thus more than $k$ vertices have to be deleted to destroy every edge in this structure.

  If there are more than $3k$ edges that intersect in two vertices, either more than $2k$ edges correspond case one or more than $k$ to case two and thus the graph cannot be solved with $k$ vertex deletions.
\end{proof}

\begin{rrule}[High Degree for 4-Edges]
  \label{rrule:highdegree4}
  If there is a vertex that has more than $3k^2$ incident 4-edges, remove it from the graph, remove any incident edges and reduce $k$ by one.
\end{rrule}

\begin{lemma}
  \label{lem:highdegree4}
  \autoref{rrule:highdegree4} is correct and can be carried out in time $\bigO(n + m)$. A graph has at most $3k^3$ 4-edges, if this rule cannot be applied to it and it is solvable with $k$ vertex deletions.
\end{lemma}

The proof is analogous to the proof for \autoref{lem:highdegree3} by substituting \autoref{lem:smallhand} with \autoref{lem:inducedclique}.

\begin{theorem}[Problem Kernel for \grmd]
  \grmd{} has a problem kernel with at most $9k^3 + 2k^2 + k$ marks. The characteristic hypergraph of the ruler of a kernelized instance has at most $3k^3$ 4-edges and $3k^2$ 3-edges. The kernel can be computed in $\bigO(kn + km)$~time if the characteristic hypergraph is known.
\end{theorem}

\begin{proof}
  To compute the kernel, proceed as follows: Apply \autoref{rrule:highdegree3}, apply \autoref{rrule:highdegree4} and recurse until neither applies anymore. Then apply \autoref{rrule:lonevertices} until it does not apply anymore.
  
  Since both high-degree reduction rules can be applied at most $k$ times, the procedure recurs at most $k$ times and the overall running time adds up to $\bigO(kn + km)$. \autoref{rrule:lonevertices} can of course be applied exhaustively in $\bigO(n)$~time.

% The condition of \autoref{lem:inducedclique} can be checked in time $\bigO(n^2 + m)$: One can use an array of size $\bigO(n^2)$, indexed by every unordered pair of vertices and initialized with zero. Then iterate over every 4-edge present in the graph and increase the array entry of any unordered pair of vertices in the edge by one. It follows that the overall running time is bounded by a term in $\bigO(kn^2 + km)$.

  The upper bound on the 3- and 4-edges follows from \autoref{lem:highdegree3} and \autoref{lem:highdegree4}. In a yes-instance there is a set $S$ of at most $k$ vertices, such that every edge in the graph has a non-empty intersection with $S$. That means in each of the 4-edges, there are at most three vertices not in $S$ and thus there are at most $9k^3 + k$ vertices in 4-edges. This argument holds analogously for 3-edges and thus there are at most $9k^3 + 2k^2 + k$ vertices in a yes-instance.
\end{proof}

\subsection{Mark Preservation Parameter}

We also tried to analyze \grsm{} with respect to the size of the sought \gr. However, the problem mostly escaped our attempts. We conjecture it to be W[1]-hard, and gather some hints towards this in the following. 

With the notion of hypergraph characterization, \grsm{} reformulates as follows.
\probsp{}
\decprob{\grsm}{A ruler $R \subseteq \mathbb{N}$ and $k\in\mathbb{N}$.}{Is there a ruler $R' \subseteq R$ with size at least $k$, such that $H_{R'}$ contains no edges?}
\probsp{}
At first, observe that \grsm{} can be solved with \textsc{Independent Set} on hypergraphs. In \textsc{Independent Set} on hypergraphs, a hypergraph $H$ is given and a (maximum size) vertex set $S$ is sought, such that $H[S]$ contains no edges. However, \textsc{Independent Set} on $r$-uniform hypergraphs, parameterized with the size of the sought vertex set, has been proven to be W[1]-hard by \citet{NR07}. Their proof relies on heavily overlapping edges, and thus this approach is not directly applicable to our problem. Nevertheless, our problem retains some of the features of general \textsc{Independent Set}. For example, a simple branching strategy like the one in \textsc{Hitting Set} seems not to be feasible, because in our instances too, there can be edges that do not contain solution vertices at all.

If we modify the problem slightly, it indeed becomes W[1]-hard:
\probsp{}
\decprob{Golomb Subruler $\geq$ Double Conflicts}{A ruler $R \subseteq \mathbb{N}$ and $k\in\mathbb{N}$.}{Is there a ruler $R' \subseteq R$, such that $H_{R'}$ contains no edges and $R'$ has a number of marks that is least $k$ plus two times the number of edges in $H_R$?}
\probsp{}
The proof for this is given in \autoref{sec:simplnphardness} on page \pageref{anch:grsm'w-hard}.

\chapter{Implementation and Empirical Results}
\label{sec:impl}

In this chapter, we investigate the practical implications of our considerations in \secref{sec:fpt}. We implemented the search tree algorithm that is discussed there. It uses some heuristic improvements and the problem kernel we observed. We first describe the implementation in detail and then report on our results. Some of the heuristic improvements prove very effective in practice.

\section{Description of the Implementation}

\paragraph{Problem Definition.} Our implementation solves the following problem with help of the fixed-parameter algorithm we developed in \autoref{sec:markdeletion}.
\optprob{\textsc{Maximum Mark Golomb Subruler}}{A ruler $R \subset \mathbb{N}$.}{Find a \gr{} $R' \subseteq R$ of maximum cardinality.}
Observe that, theoretically, this problem can be used to answer a variety of questions: By using appropriate rulers---with marks taken consecutively from~$\mathbb{N}$---the following problem reduces to the one above.
\optprobnotitle{}{An integer $D \in \mathbb{N}$.}{Find a \gr{} of length at most $D$ and maximum number of marks.}
\grd{} can be solved with the second problem, which means that our algorithm can be used to directly search for optimal \gr s as well as a subroutine in other programs that produce partial instances corresponding to \textsc{Maximum Mark Golomb Subruler}.

\paragraph{Algorithm.} The algorithm first generates the characteristic hypergraph of the input ruler $R$. It then proceeds to heuristically compute a minimal subset $S$ of $R$ such that $H_{R \setminus S}$ contains no edges. This solution is used as an upper bound for the parameter in the corresponding instance of \grmd{}. We interpret this upper bound and the input ruler as an instance of \grmd{}, solve it, and successively decrease the parameter, thus obtaining successive smaller mark deletion sets, until no solution can be found anymore. The last solution~$S$ that could be found is a minimum-size set of marks such that $H_R$ contains no edges. It is clear that $R \setminus S$ must be optimal to \textsc{Maximum Mark Golomb Subruler}. See also \autoref{alg:golombsubruler}. 
\begin{algorithm}%[H]
  \KwIn{A ruler $R \subset \mathbb{N}$.}
  \KwOut{A \gr{} $R' \subseteq R$ with maximum number of marks.}
  
  \SetKwFunction{HGC}{HypergraphConstructionImproved}
  \SetKwFunction{SolveHeuristic}{SolveHeuristic}
  \SetKwFunction{SolveParameterized}{SolveParameterized}
  \SetKwData{Continue}{continue}
  \SetKwData{HeuristicSol}{MinimalMarkDeletionSet}
  \SetKwData{ParameterizedSol}{AtMost$k$MarkDeletionSet}
  \SetKwData{ParameterizedSolp}{ParameterizedSolution}
  \SetKw{False}{false}
  
  $H_R \leftarrow$ \HGC{R}\;
  \HeuristicSol $\leftarrow$ \SolveHeuristic{$H_R$}\;
  \Continue $\leftarrow$ true\;
  $k \leftarrow |\text{\HeuristicSol}|$\;
  \ParameterizedSol $\leftarrow$ \HeuristicSol\;
  \While{\Continue}{
    \ParameterizedSolp $\leftarrow$ \SolveParameterized{$H_R$, $k$, $\emptyset$, $\emptyset$}\;
    \lIf{\ParameterizedSolp = \False}{\Continue $\leftarrow$ \False\;}
    \Else{
      $k \leftarrow |\text{\ParameterizedSolp}| - 1$\; 
      \ParameterizedSol $\leftarrow$ \ParameterizedSolp\;
    }
  }
  $R' \leftarrow R \setminus \text{\ParameterizedSol}$\;
  \Return $R'$\;
  \SetAlgoRefName{Find\-Golomb\-Sub\-ruler}
  \caption{Computing Golomb subruler with maximum number of marks for a given ruler}
 \label{alg:golombsubruler}
\end{algorithm}

The heuristic solution is computed via two very simple strategies and the best solution is kept. The first strategy greedily deletes vertices from the input graph $H_R$, until it is edge-free. The second strategy greedily selects an edge, deletes all vertices of this edge, and iterates until the graph contains no edges anymore. Observe that the second strategy yields a solution that contains at most four times the number of vertices in an optimal solution, because in any edge at least one vertex has to be deleted.

For the subroutine \ref{alg:solvepara} that solves the instances of \grmd{}, we make use of the (heuristic) improvements described in \autoref{sec:mdalgorithm}. That is, we cementate vertices (see page \pageref{anch:cementating}), if the recursive calls of \ref{alg:solvepara} return without finding a solution, we use the implications of the cementated vertices and we try to find edges that imply favorable branching vectors. Additionally, in search tree nodes, we apply the reduction rules we described in \autoref{sec:mdkernel}. A pseudocode description of the search tree subroutine can be seen in \autoref{alg:solvepara}. (The cementating of marks in lines 10 through 14 is simplified for readability.)

\begin{algorithm}%[H]
  \LinesNumbered

  \KwIn{A characteristic Hypergraph $H_R$, an integer $k$, a set of cementated marks $C \subset R$ and a solution set $S$.}
  \KwOut{A set of marks $S$ such that $|S| \leq k$ and $H_{R \setminus S}$ contains no edges or false if such a set does not exist.}

  \SetKwFunction{SolveHeuristic}{SolveHeuristic}
  \SetKwFunction{SolveParameterized}{SolveParameterized}
  \SetKwData{Continue}{continue}
  \SetKwData{BranchingEdge}{BranchingEdge}
  \SetKwData{Sol}{Solution}
  \SetKw{False}{false}

  \lIf{$k < 0$}{\Return \False\;}
  \lIf{$H_R$ contains no edges}{\Return $S$\;}
  Apply the problem kernel to $H_R$ and $k$, adding deleted vertices to $S$\;
  \lIf{$k < 0$ or $H_R$ exceeds the problem kernel size}{\Return \False\;}
  \lIf{$H_R$ contains no edges}{\Return $S$\;}
  \BranchingEdge $\leftarrow$ An edge $e$ in $H_R$ such that $e \setminus C$ has minimum size\;
  \For{$i \in \text{\BranchingEdge} \setminus C$}{
    \Sol $\leftarrow$ \SolveParameterized{$H_{R \setminus \{i\}}$, $k - 1$, $C$, $S \cup \{i\}$}\;
    \uIf{\Sol = \False}{
      $C \leftarrow C \cup \{i\}$\;
      \For{$a, b, c \in C$}{
        $X \leftarrow \{c + |a - b|, c - |a - b|\}$\;
        \lIf{$X \cap C \neq \emptyset$}{\Return \False\;}
        $H_R \leftarrow H_{R \setminus X}$; $S \leftarrow S \cup (X \cap R)$; $k \leftarrow k - |X \cap R|$\;
    }}
    \lElse{\Return \Sol\;}
    
  }
  \Return \False\;

  \SetAlgoRefName{Solve\-Parameterized}
  \caption{Solving \grmd}
 \label{alg:solvepara}
\end{algorithm}

\paragraph{Implementation.} We implemented the above described algorithm in the \emph{Objective Caml} language. %TODO cite
Objective Caml is a multi-paradigm language, allowing for object oriented, functional and procedural programming styles. We chose it because first, it allows for easy transition from theoretical algorithms to practical programs via the functional programming possibilities. Second, because it is very robust to programming errors through strict variable typing and automatic type inference. And third, because the implementation is meant to be a proof of concept and not a highly optimized solver. 

The design is a simple one-process, one-threaded solution.

\paragraph{Testing Environment.} The experiments were conducted on an Intel Xeon E5410 machine with 2.33 GHz, 6 MB L2 cache, and 32 GB main memory. The operating system was GNU/Linux with a kernel of version 2.6.26. The Objective Caml compiler used was of version 3.10.2. We did not utilize the symmetric multiprocessing capabilities available on the system.

\section{Running Times and other Results}

We conducted our experiments on rulers with $n$ marks of the form $\{i \in \mathbb{N} : i < n\}$. 

\paragraph{General Running Times.} Varying the number of marks $n$, we observe the running time behavior shown in \autoref{fig:runningtimesgeneral}. The running times obviously increase exponentially with increasing $n$. However, if we fit the running time function~$f(n) = a + b n^2 c^n$ to the empirical running times via the Levenberg-Marquardt algorithm \cite{Mo78}, we get the following result:

$$f(n) = 0.989502 + 9.89765 \times 10 ^{-5} n^2 1.18727^n$$

Even with iteratively using the $\bigO^*(4^k)$ subroutine, we still obtain an algorithm whose empirical running time has an exponential term below $1.2^n$. Although this seems like a good result, it would be necessary to examine rulers of about $550$ marks, to proof or disproof the optimality of the currently shortest known \gr{} with $27$ marks (as of now, \gr s of minimum length with up to 26 marks are known). Assuming the running time behavior corresponds to the function~$f$ we determined above, this would result in running times in the order of $10^{30}$~years on our architecture, i.e., this particular implementation is not feasible for this application. However, with the reduction rules and heuristic improvements, it still might be possible to use it as subroutine in such search algorithms. Experiments on random rulers as input could settle this question.

\begin{figure}
  \begin{center}
    \includegraphics{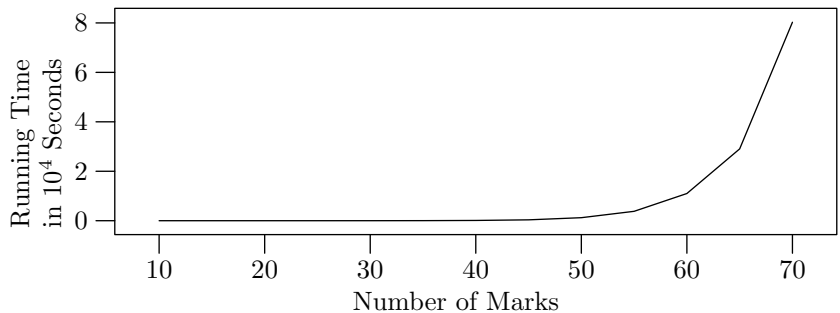}
    \caption{Observed running times of \autoref{alg:golombsubruler} versus number of marks of the input ruler.}
    \label{fig:runningtimesgeneral}
  \end{center}
\end{figure}

\paragraph{Greedy Strategies.} In \autoref{alg:golombsubruler}, two greedy strategies are used to obtain an upper bound for the parameter. We observe that in our experiments the strategy that greedily deletes vertices always yields a solution that is superior to greedily deleting all vertices of one edge.
The size of the greedy solution is very close to the optimal solution in our experiments, as shown in \autoref{fig:greedyvsopt}. In this context, it would be interesting to test whether there are lower bounding techniques that can be applied fast.

\begin{figure}
  \begin{center}
    \includegraphics{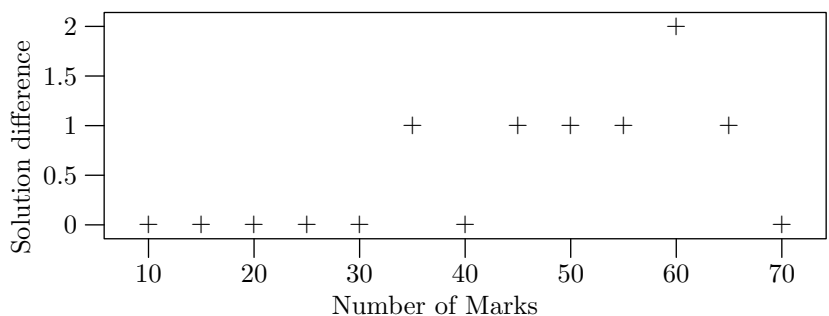}
    \caption{Greedy solution size minus optimal mark deletion set size. The greedy algorithm performs very close to the optimal solution sizes.}
    \label{fig:greedyvsopt}
  \end{center}
\end{figure}

\begin{table}
  \centering
    \begin{tabularx}{\textwidth}{X XcX XrX XrX X}
      \toprule
      &\multicolumn{3}{c}{Number of Marks} 
      & \multicolumn{3}{c}{No Cementation} 
      & \multicolumn{3}{c}{Cementation} &\\
      \midrule
      &&20 &&& 35.22 &&& 0.18 && \\
      &&25 &&& 8313.78 &&& 1.27 && \\
      \bottomrule
    \end{tabularx}
  \caption{Cementating vertices yields tremendous speedups. Running times in seconds.}
  \label{tab:cementspeedup}
\end{table}

\begin{table}
  \centering
    \begin{tabularx}{\textwidth}{X XcX XrX XrX X}
      \toprule
      &\multicolumn{3}{c}{Number of Marks}
      & \multicolumn{3}{c}{Search Tree Nodes}
      & \multicolumn{3}{c}{Vertices Deleted} &\\ % TODO & Branches Aborted  \\
      \midrule
      &&\hspace{0.5cm}20\hspace{0.5cm} &&& 2,734 &&& 3,278&& \\
      &&40 &&& 893,070 &&& 1,351,977&& \\
      &&60 &&& 60,035,055 &&& 103,541,222&& \\
      \bottomrule
    \end{tabularx}
  \caption{Search tree nodes and vertices deleted due to cementation in comparison.}
  \label{tab:cementdelete}
\end{table}

%\pagebreak
\paragraph{Cementating Vertices.} The cementating of vertices that have a recursive call returning without solution yields large speedups as is shown in \autoref{tab:cementspeedup}. In \autoref{tab:cementdelete}, it is shown, how many vertices are deleted by it in the course of the search tree, when varying the size of the input ruler. This success can be explained with the following example:

Suppose branching has been done on edges that contain vertices of a confined region on the input ruler. A set of $d$ cementated vertices have been gathered and now the algorithm moves on to branch on an edge in a different region on the ruler. If the first recursive call for a vertex in the edge returns negative, $d (d - 1) /2$ vertices can be deleted, because the vertex now is cementated and the set of cementated vertices must be a \gr . The number of deleted vertices most likely increases another time, when the recursive call for the second vertex in the edge returns negative and so on. This means that we get a branching vector of $(1, d(d-1)/2, d(d-1)/2 + c, ...)$ for this particular case. The running times suggest that the branching vector of an edge is likely to be much more favorable in practise with cementated vertices.

\paragraph{Reduction Rules in the Search Tree.} 
When not using cementation, we obtain speedups of about two using the high-degree rules. However, when using cementation they do not yield running time benefits anymore; although the reduction rules delete two to three vertices per search tree node and are responsible for many direct terminations of branches (as we show in \autoref{tab:reddelete}). The time needed to calculate the rules outweighs their benefits in our instances. This is shown for the high-degree rule in \autoref{tab:highdegree}. This can not be remedied by only applying the reduction rules every $d$'th search tree node in our instances. However, considering the number of deleted vertices and aborted branches and the fact that they seem to improve in overall effectiveness with increasing number of vertices (see \autoref{tab:reddelete}), it might prove worthwhile to thoroughly optimize the implementation to support these routines in minimal time.
%However, if we restrict the evaluation of the reduction rules to every $d$'th search tree node that is evaluated, we again obtain running time benefits, as is shown in TODO fig.

\begin{table}
  \centering
  \begin{tabularx}{\textwidth}{XcX XrX XrX XrX}
    \toprule
     \multicolumn{3}{c}{Number}
    & \multicolumn{3}{c}{}
    & \multicolumn{3}{c}{}
    & \multicolumn{3}{c}{}  \\ 
     \multicolumn{3}{c}{of Marks}
    & \multicolumn{3}{c}{Search Tree Nodes}
    & \multicolumn{3}{c}{Vertices Deleted}
    & \multicolumn{3}{c}{Branches Aborted}  \\ 
    \midrule
    &20 &&& 2,734 &&& 6,148 &&& 1,856 & \\
    & 40 &&& 893,070 &&& 2,472,270 &&& 570,670 &\\
    & 60 &&& 60,035,055 &&& 186,900,562 &&& 34,506,053 &\\
    \bottomrule
  \end{tabularx}
  \caption{Search tree nodes compared to vertices deleted and branches aborted due to the high-degree, leaf-edge and lone-edge reduction rules.}
  \label{tab:reddelete}
\end{table}
\begin{table}
  \centering
  \begin{tabularx}{\textwidth}{X XcX XrX XrX X}
    \toprule
    & \multicolumn{3}{c}{Number}
    & \multicolumn{3}{c}{With}
    & \multicolumn{3}{c}{Without} & \\ 
    & \multicolumn{3}{c}{of Marks}
    & \multicolumn{3}{c}{High-Degree Rules}
    & \multicolumn{3}{c}{High-Degree Rules} & \\ 
    \midrule
    &&20 &&& 0.18 &&& 0.17 &&  \\
    &&30 &&& 6.07 &&& 5.48 && \\
    &&40 &&& 106.21 &&& 97.79 && \\
    &&50 &&& 1212.30 &&& 1134.45 && \\
    \bottomrule
  \end{tabularx}
  \caption{Running times using and ignoring the high-degree reduction rules whilst cementating vertices. All values in seconds. }
  \label{tab:highdegree}
\end{table}

\paragraph{Conclusion} On the plus side, we could show that our heuristic improvements are very effective and our reduction rules delete many vertices and are effective in trimming the search tree. Unfortunately, despite a relatively small exponential term in the observed running times, this did not lead to an implementation that is feasible for discovering new \gr s.

% \begin{figure}
%   \begin{center}
%     \includegraphics{highdegree.1}
%     \caption{Improving the computation time/benefit tradeoff of the reduction rules.}
%     \label{fig:highdegree}
%   \end{center}
% \end{figure}

% TODO

% \begin{itemize}
%   \item running times on succs, random?
%   \item application of the branching rules
%   \item comparison with distributed.net if possible (TODO...)
% \end{itemize}

\chapter{Conclusion and Outlook}

\paragraph{Our Work.} In this work, we have given a short overview of some hurdles and algorithmic possibilities in the field of \gr s. The basic problems \gro{} and \grd{} seem to be elusive to classic complexity classification and may well lie between P and NP. To settle these questions would of course be very interesting; however, it seems likely that this would imply major new insights into classic complexity theory and/or number theory.

The natural hypergraph characterization we have given for rulers makes it possible to get a better understanding of conflicts with respect to \gr s. We observed some structure in characteristic hypergraphs that we later exploited for a problem kernel. This implies that a more sophisticated structurization of those graphs could lead to much better algorithms and thus would be very interesting. A more thorough understanding of the graphs could also be used to settle other questions related to \gr{} construction, for example the W[1]-hardness of \grsm{} with respect to the mark preservation parameter.

We have given a simplified proof for the NP-hardness of \grsm{} that also lead to a W[1]-hardness result for a modified problem that did not directly follow from the original proof. We hope that this makes fixed-parameter research for related problems more accessible and attractive.

Concerning fixed-parameter algorithmics, we have observed that \grmd{} is tractable and we have given a cubic-size problem kernel. We strongly believe that this bound can be surpassed and we also have just scratched the surface regarding solution algorithms.

The implementation of a corresponding algorithm showed that the reduction rules prove quantitatively efficient in a search tree, however, their effective running times have to be improved. The technique of cementating vertices applied to \gr s proved very effective in practice. %Perhaps there are more such heuristic strategies that yield such speedups.

%TODO implementation

\paragraph{Other Aspects.} There also are some interesting aspects which we did not cover in this work. We briefly list some of them here:

$\bullet$ To our knowledge the relations between \textsc{Difference Cover}, \textsc{Turnpike} and \gro{} have mostly been left unexplored in the past. There might be a trinity similar to \textsc{Clique}, \textsc{Vertex Cover} and \textsc{Independent Set}.

%$\bullet$ The bound $G(n) \leq n^2$ as stated by Erd\H{o}s originally in relation to the Sidon Set Problem to our knowledge is still an open question.  %TODO zitat
%While the computational proof of the conjecture for $n \leq 65000$ by \citet{Dim02} may be a hint to the truthfulness of it, other conjectures have been disproven with very large counterexamples in the past (eg. P\'olya and Mertens conjectures). %TODO zitat

$\bullet$ Also to our knowledge the parameterized complexity of \grd{} has not been touched yet. Since easily checkable lower bounds and allegedly efficiently constructible upper bounds on the length of a \gr{} exist \cite{Dim02}, it may be possible to apply a below or above guarantee parameterization. 

$\bullet$ There is a problem called \textsc{Minimum Redundancy Linear Array} \cite{Mof68}, where rulers are sought which are perfect---they measure every distance up to their length---and have minimum redundant differences between marks. Since in a \gr{} there is no redundancy at all, maybe this would make for a practical parameterization. %TODO zitat

$\bullet$ Considering the research in local search and evolutionary strategies \cite{CF05, CDFH07, PTC03, SHL95, TPC07}, it might be interesting to contemplate fixed-parameter tractability for problems arising in this field.

$\bullet$ The notion of treewidth is a popular topic in fixed-parameter algorithms. There is a theorem by Courcelle \cite{FG06} that classifies every decision problem that can be formulated in monadic second-order logic as fixed-parameter tractable with respect to the parameter treewidth. One could investigate if such formulations exist for the hypergraph problems we considered in this work.
%TODO zitat

%generelles ergebnis von sepp anwendbar?
%parallelisierung unseres ansatzes?

%parametrized Complexity
%  Below/above bound
%  max dist consecutive elements
%  \# missed distances
%  \# allowed multiple instances of distances -> application minimum redundancy linear arrays.

%$\bullet$ Maybe the link to gracefully labeled graphs is more thorough and golomb rulers imply graceful labelings for other graph structures.

\bibliographystyle{abbrvnat}
\bibliography{seminar.bib}
 
\end{document}